\documentclass[reqno,12pt]{amsart}

\usepackage{color} 
\usepackage{ifpdf}
\ifpdf
    \usepackage[pdftex]{graphicx}
    \usepackage[pdftex]{hyperref}
    \hypersetup{
        unicode=false,          
        pdftoolbar=true,        
        pdfmenubar=true,        
        pdffitwindow=false,     
        pdfstartview={FitH},    
        pdftitle={MCP Article},      
        pdfauthor={Michael Holst},   
        pdfsubject={Mathematics},    
        pdfcreator={Michael Holst},  
        pdfproducer={Michael Holst}, 
        pdfkeywords={PDE, analysis, mathematical physics}, 
        pdfnewwindow=true,      
        colorlinks=true,        
        linkcolor=red,          
        citecolor=blue,         
        filecolor=magenta,      
        urlcolor=cyan           
    }

\else
    \usepackage{graphicx}
    \usepackage{pstricks}
    
    \newcommand{\href}[2]{#2}
\fi

\usepackage{times}
\usepackage{amsfonts}

\usepackage{amsmath}
\usepackage{amsthm}
\usepackage{amssymb}
\usepackage{amsbsy}
\usepackage{amscd}

\usepackage{enumerate}
\usepackage{verbatim}
\usepackage{subfigure}

\newtheorem{theorem}{Theorem}[section]
\newtheorem{corollary}[theorem]{Corollary}

\newtheorem{proposition}[theorem]{Proposition}

\newtheorem{definition}[theorem]{Definition}

\newtheorem{remark}[theorem]{Remark}

\numberwithin{equation}{section}  

  \newcounter{mnote}
  \let\oldmarginpar\marginpar
    \renewcommand\marginpar[1]{\-\oldmarginpar[\raggedleft\footnotesize #1]%
    {\raggedright\footnotesize #1}}

\definecolor{myblue}{rgb}{0.2,0.2,0.7}
\definecolor{mygreen}{rgb}{0,0.6,0}
\definecolor{mycyan}{rgb}{0,0.6,0.6}
\definecolor{myred}{rgb}{0.9,0.2,0.2}
\definecolor{mymagenta}{rgb}{0.9,0.2,0.9}
\definecolor{mywhite}{rgb}{1.0,1.0,1.0}
\definecolor{myblack}{rgb}{0.0,0.0,0.0}

\newcommand{\beq}{\begin{equation}}
\newcommand{\eeq}{\end{equation}}
\newcommand{\beqa}{\begin{eqnarray}}
\newcommand{\eeqa}{\end{eqnarray}}
\newcommand{\g}{\mbox{\textsf{g}}}    

\newcommand{\cH}{{\mathcal H}}

\newcommand{\cL}{{\mathcal L}}
\newcommand{\cM}{{\mathcal M}}

\newcommand{\cT}{{\mathcal T}}

\newcommand{\bg}{{\bf g}}

\newcommand{\bj}{{\bf j}}

\newcommand{\bv}{{\bf v}}
\newcommand{\bw}{{\bf w}}

\def\ee{\epsilon}

\def\la{\lambda}

\def\La{\Lambda}

\newcommand{\ol}{\overline}

\begin{document}

\title{Non-uniqueness of solutions to the conformal formulation}

\author[M. Holst]{Michael Holst}
\email{mholst@math.ucsd.edu}

\author[C. Meier]{Caleb Meier}
\email{meiercaleb@gmail.com}

\address{Department of Mathematics\\
         University of California San Diego\\ 
         La Jolla CA 92093}

\thanks{MH was supported in part by NSF Awards~0715146 and 0915220,
and by DOD/DTRA Award HDTRA-09-1-0036.}
\thanks{CM was supported in part by NSF Award~0715146.}

\date{\today}


\begin{abstract}
It is well-known that solutions to the conformal formulation of the 
Einstein constraint equations are unique in the cases of
constant mean curvature (CMC) and near constant mean curvature (near-CMC).
However, the new far-from-constant mean curvature (far-from-CMC) 
existence results due to Holst, Nagy, and Tsogtgerel in 2008, 
to Maxwell in 2009, and to Dahl, Gicquaud and Humbert in 2010, 
are based on degree theory rather than on the (uniqueness-providing)
contraction arguments that had been used for all non-CMC existence 
results prior to 2008.
In fact, Maxwell demonstrated in 2011 that solutions are 
non-unique in the far-from-CMC case for certain types of low-regularity
mean curvature.
In this article, we investigate uniqueness properties of solutions to the 
Einstein constraint equations on closed manifolds using tools from 
bifurcation theory.
For positive, constant scalar curvature and constant mean curvature, 
we first demonstrate existence of a critical energy density for the 
Hamiltonian constraint with unscaled matter sources.
We then show that for this choice of energy density, the linearization of the 
elliptic system develops a one-dimensional kernel in both the CMC and 
non-CMC (near and far) cases.
Using Liapunov-Schmidt reduction and standard tools from nonlinear analysis,
we demonstrate that solutions to the conformal formulation with unscaled 
data are non-unique by determining an explicit solution curve, and by
analyzing its behavior in the neighborhood of a particular solution.  
\end{abstract}

\maketitle


\vspace*{-1.2cm}
{\scriptsize
\tableofcontents
}

\section{Introduction}

In this paper we demonstrate that solutions to the
Einstein constraint equations on a 3-dimensional closed manifold $(\cM,\hat{g}_{ab})$ with no 
conformal killing field are non-unique.
More specifically, we show that solutions to the conformal
formulation of the constraint equations with an unscaled matter source on $(\cM,\hat{g})$ exhibit non-uniqueness in the case that the scalar curvature is positive and constant.
Letting $\hat{k}_{ab}$ be a $(0,2)$ tensor and $\hat{R}$ and $\hat{D}$
be the scalar curvature and connection associated with $\hat{g}_{ab}$, the constraint equations take
the form
\begin{align}
&\hat{R}+\hat{k}^2-\hat{k}^{ab}\hat{k}_{ab} = 2\kappa\hat{\rho},\label{eq1:3july12} \\
&\hat{D}^a\hat{k} +\hat{D}_b\hat{k}^{ab} +\kappa \hat{j}^a = 0. \label{eq2:3july12}
\end{align}
Equation \eqref{eq1:3july12} is known as the {\bf Hamiltonian Constraint}
and \eqref{eq2:3july12} is known as the {\bf momentum constraint}.

Equations \eqref{eq1:3july12} and \eqref{eq2:3july12} form a system of coupled elliptic partial differential equations.
When one attempts
to solve the constraint equations they are faced with the problem of having twelve pieces of initial data and only four constraints.
One solution to this problem is to attempt to parametrize solutions to \eqref{eq1:3july12} and \eqref{eq2:3july12} by formulating the constraints so that eight pieces of initial data are freely specifiable while four are determined by \eqref{eq1:3july12}-\eqref{eq2:3july12}.
The conformal transverse traceless (CTT) decomposition and the conformal thin sandwich method (CTS method) are standard ways of doing this.
The extended conformal thin sandwich method (XCTS method ) is popular among numerical relativists and
reformulates \eqref{eq1:3july12} and \eqref{eq2:3july12} as a coupled system of 5 elliptic equations.
In the CTT method one decomposes $\hat{k}_{ab}$ into its trace or mean curvature and trace free part and then scales this trace free tensor, the metric $\hat{g}_{ab}$ and the source terms $\hat{\rho}$ and $\hat{j}$ by judicious choices of some power of a positive, smooth function $\phi$.
The choice of scaling power for each term is typically made to simplify the analysis of the resulting system.
In particular, one chooses powers to eliminate terms involving $(D_a\phi)/ \phi$ and so that the system decouples when the mean curvature is constant.

It is well-known that solutions to the CTT formulation of the constraint equations with scaled data sources are unique in the event that the mean curvature is constant (known as the ``CMC case''), 
or near constant (the ``near-CMC case'');
cf.~\cite{JI95,IM96,ACI08,HNT07a,HNT07b}.
Prior to 2008, all non-CMC existence results were only possibly for the near-CMC case, and were established using contraction arguments, which provided uniqueness for free once existence was established.
However, beginning in 2008 with the first true ``far-from-CMC''
(the non-CMC case without near-CMC restrictions) existence result 
in~\cite{HNT07a}, all far-from-CMC results to 
date~\cite{HNT07a,HNT07b,M09} are based on a variation of the 
Schauder Fixed-Point Theorem.
This also includes the more recent work~\cite{DGH10a}, 
which uses the Schauder framework from~\cite{HNT07a,HNT07b,M09} 
as part of a pseudo-variational argument.
As a result, little is known about uniqueness of far-from-CMC solutions.
In fact, in 2011 Maxwell demonstrated in~\cite{MaD11} that solutions of the 
CTT formulation of the constraint equations are non-unique in the far-from-CMC
case for certain families of low regularity mean curvatures.
However, as noted by Maxwell in \cite{MaD11}, the discontinuous mean 
curvature functions considered by Maxwell in \cite{MaD11} fall outside of the
best existing non-CMC rough solution theory established in~\cite{HNT07b}.

In \cite{PY05}, Pfeiffer and York provided numerical evidence for non-uniqueness of the XCTS method on an asymptotically Euclidean manifold.
In \cite{BOP07}, Baumgarte, O'Murchadha, and Pfeiffer conjectured that the non-uniqueness demonstrated by Pfeiffer and York was related to the fact
that certain terms in the momentum constraint related to the lapse function have the ``wrong sign", which prevents an application
of the maximum principle.
To support their claim, the authors of \cite{BOP07} analyzed a simplified system corresponding to a spherically symmetric constant density star and explicitly constructed two branches of solutions.
In their analysis they proved that solutions to the
Hamiltonian constraint \eqref{eq1:3july12} with an unscaled matter source are non-unique.
Then in \cite{DW07}, Walsh generalized the work in~\cite{BOP07} by applying a Liapunov-Schmidt reduction to both the Hamiltonian constraint with an unscaled matter source and
to the XCTS system on an asymptotically Euclidean manifold.
However, Walsh relied on the assumption of the existence of a critical density for which the linearization of these two systems developed a one-dimensional kernel.
Here we extend the work of Walsh by applying a Liapunov-Schmidt reduction to the CTT formulation of the constraint equations on a closed manifold.
We explicitly construct a critical, constant density in the event that the scalar curvature is positive and constant and the transverse traceless tensor has constant magnitude.
For this particular density, we then show that solutions to the CTT formulation with an unscaled density are non-unique.

As in \cite{BOP07,DW07}, we consider a less standard conformal formulation of the constraints by allowing unscaled matter sources $\rho$ and $\bj$.
However, as opposed to considering the CTS and XCTS formulations
as in \cite{PY05,BOP07,DW07}, we consider the CTT formulation.
By decomposing our initial data
\begin{align}
\hat{k}_{ab} = \hat{l}_{ab}+\frac13\hat{g}_{ab}\hat{\tau},
\end{align}
where $\hat{\tau}  = \hat{k}_{ab}\hat{g}^{ab}$ is the trace and $\hat{l}_{ab}$ is the traceless part,
making the following conformal rescaling
\begin{align}
\hat{g}_{ab} = \phi^4g_{ab}, \quad \hat{l}_{ab} = \phi^{-10}l^{ab}, \quad \hat{\tau} = \tau, 
\end{align}  
and then decomposing
\begin{align}
l_{ab} = (\sigma_{ab}+(\cL \bw)_{ab}),
\end{align}
where $D_a\sigma^{ab} = 0$ and 
$$
(\cL \bw)^{ab} = D^aw^b + D^bw^a-\frac23(D_cw^c)g^{ab}
$$
is the {\bf conformal Killing operator},
we obtain the following unscaled conformal reformulation of \eqref{eq1:3july12} and \eqref{eq2:3july12} that
we will analyze
\begin{align}\label{eq3:3july12}
-\Delta \phi + &\frac18R\phi +\frac{1}{12} \tau^2\phi^5 -  \frac18(\sigma_{ab}+(\cL \bw)_{ab})(\sigma^{ab}+(\cL\bw)^{ab})\phi^{-7}-2\pi\rho\phi^5=0,\\
&-D_b(\cL\bw)^{ab}+\frac23D^a \tau \phi^6 + \kappa j^a\phi^{10}=0.\nonumber
\end{align}

Our non-uniqueness results for \eqref{eq3:3july12} are of interest for a number or reasons.
Most immediately, our analysis shows that the formulation \eqref{eq3:3july12} is unfavorable due to the non-uniqueness of solutions.
Therefore, for a given system, if the CTT formulation with a scaled matter source leads to a set of constraints that is suitable for analysis, which it usually does, then one should use the scaled formulation.
However, it is not always the case that the conformal formulation with scaled matter sources is the ideal formulation for a given source.
In the case of the Einstein-scalar field system, the conformal formulation that is most amenable to analysis takes on a form very similar to the system \eqref{eq3:3july12} \cite{CIP07}.
In addition, it is the hope of the authors that these results will provide additional insight into the non-uniqueness phenomena associated with the CTT formulation  in the far-from-CMC case \cite{MaD11} and with the non-uniqueness phenomena analyzed by Pfeiffer and York~\cite{PY05}, by Walsh~\cite{DW07} and by Baumgarte, O'Murchadha, and Pfeiffer~\cite{BOP07}.
In particular, the analysis conducted in this article clearly demonstrates the effect that terms with the ``wrong sign", as discussed in~\cite{BOP07}, have on the non-uniqueness of the conformal formulations of the constraints.
In the case of \eqref{eq3:3july12}, the negative sign in front of the term $2\pi\rho\phi^5$ is undesirable given that it prevents the semilinear portion of the Hamiltonian constraint from being monotone and the corresponding energy from being convex.
By a maximum principle argument, we will see in section~\ref{sec3:14july12} that it is this term that directly contributes to the non-uniqueness properties of \eqref{eq3:3july12}.

The rest of this paper is organized as follows.
In section~\ref{sec1:14july12} we introduce the function spaces that we will use and some basic concepts from functional analysis.
Then we discuss the Liapunov-Schmidt reduction that we use to prove non-uniqueness.
The statements of the main results of this paper can be found in section~\ref{sec2:14july12}.
The remainder of this paper is then devoted to proving these results.
The foundation for our argument is developed in sections~\ref{sec3:14july12} and \ref{sec1:29jun12}.
In section~\ref{sec3:14july12} we demonstrate the existence of a critical, constant density $\rho_c$ such that 
if $g_{ab}$ has positive, constant scalar curvature, $|\sigma|$ is constant and $ \bj^a=0$, the Hamiltonian constraint in \eqref{eq3:3july12} will have a positive solution if $\rho \le \rho_c$ and will have no positive solution if $\rho > \rho_c$.
Then in section~\ref{sec1:29jun12} we use the properties of $\rho_c$ to show that there exists a function $\phi_c$ at which
the linearizations of the uncoupled Hamiltonian operator (CMC case) and coupled system (non-CMC case) have one-dimensional kernels.
The existence of a one-dimensional kernel then allows us to apply the Liapunov-Schmidt reduction in section~\ref{sec4:14july12} in the CMC case and in section~\ref{sec5:14july12} in the non-CMC case.
In particular, in section~\ref{sec4:14july12} we determine an explicit solution curve for \eqref{eq3:3july12} that goes through the point $(\phi_c,0)$ in the CMC case.
An analysis of this curve then implies the non-uniqueness of solutions to \eqref{eq3:3july12} when the mean curvature is constant.
Similarly, in section~\ref{sec5:14july12} we also determine an explicit solution curve for the full, uncoupled system \eqref{eq3:3july12} through
a point of the form $((\phi_c,{\bf 0}),0)$.
Again, an analysis of this curve reveals non-uniqueness in the event that the mean curvature is non-constant.

\section{Preliminary Material}\label{sec1:14july12}

In this section we give a brief definition of the function spaces, norms and notation that we will use in this
article and then discuss some basic concepts from functional analysis and bifurcation theory that will be
necessary going forward.

\subsection{Banach Spaces, Hilbert Spaces and Direct Sums}

We introduce the fundamental properties of the function spaces with which we will
be working.
We will primarily be working with Banach spaces, however at times we
will need to consider these spaces as subspaces of a Hilbert space.
For convenience, we present the basic definitions of these general spaces and define the direct sum of two vector spaces, which will be
necessary in our non-uniqueness analysis.

The basic space that we will be working with is a Banach space, where
a {\bf Banach space} $X$ is a complete, normed vector space.
If the norm
$\|\cdot\|$ on $X$ is induced by an inner product, we say that $X$ is a {\bf Hilbert Space}.
 One can form new Banach spaces and Hilbert spaces from preexisting spaces by considering the direct
sum.
\begin{definition}\label{def2:24july12}
Suppose that $X_1$ and $X_2$ are Banach spaces with norms $\|\cdot\|_{X_1}$ and $\|\cdot\|_{X_2}$.
Then
the direct sum $X_1\oplus X_2$ is the vector space of ordered pairs $(x,y)$ where $x\in X_1$, $y\in X_2$ and addition and scalar multiplication
are carried out componentwise.
\end{definition}
\noindent
We have the following proposition: 
\begin{proposition}\label{prop1:20aug12}
The vector space $X_1\oplus X_2$ is a Banach space when given the norm 
\begin{align}
\|(x,y)\|_{X_1\oplus X_2} =\left(\|x\|^2_{X_1}+\|y\|^2_{X_2}\right)^{\frac12}.
\end{align}
\end{proposition}
\begin{proof}
This follows from the fact that $\|\cdot\|_{X_1}$ and $\|\cdot\|_{X_2}$ are norms and
the spaces $X_1$ and $X_2$ are complete with respect to these norms.
\end{proof}
\noindent
We have a similar proposition for Hilbert spaces.
\begin{proposition}\label{def1:24july12}
Suppose that $\cH_1$ and $\cH_2$ are Hilbert spaces with inner products $\langle \cdot,\cdot \rangle_{\cH_1} $ and $\langle \cdot,\cdot \rangle_{\cH_2} $.
Then the direct sum $H_1\oplus H_2$ is a Hilbert space with inner product
\begin{align}
\langle (w,x),(y,z) \rangle_{\cH_1\oplus \cH_2}  =\langle w,y \rangle_{\cH_1} + \langle x,z \rangle_{\cH_2}.
\end{align}
\end{proposition}
\begin{proof}
That $\langle \cdot,\cdot \rangle_{\cH_1\oplus \cH_2}$ is an inner product follows from the fact
that $\langle \cdot,\cdot\rangle_{\cH_1} $ and $\langle \cdot,\cdot \rangle_{\cH_2}$ are inner products.
The expression 
$$\|(u,v),(u,v)\|_{\cH_1\oplus \cH_2} = \sqrt{\langle (u,v),(u,v) \rangle_{\cH_1\oplus \cH_2} },$$
is a norm on $\cH_1\oplus \cH_2$ that coincides with the norm in Proposition~\ref{prop1:20aug12}
in the event that the norms on $X_1$ and $X_2$ are induced by inner products.
\end{proof}
\noindent
See \cite{EZ86} for a more complete discussion about the direct sums of Banach spaces.

\subsection{Function Spaces}

Let $E$ denote a given vector bundle over
$\cM$.
In this paper
we will consider the Sobolev spaces $W^{k,p}(E)$, the space of $k$-differentiable sections $C^k(E)$,
and the H\"{o}lder spaces $C^{k,\alpha}(E)$ where $k \in \mathbb{N},~ p \ge 1, ~\alpha \in (0,1)$ and $E$ will either be the
vector bundle $\cM \times \mathbb{R}$ of scalar-valued functions or $\cT^r_s \cM$,
the space of $(r,s)$ tensors.
Note that all of these spaces with the following norm definitions are Banach spaces
and the space $W^{k,2}(E)$ is a Hilbert space for $k \in \mathbb{N}$.

Fix a smooth background metric $g_{ab}$ and let $v^{a_1,\cdots,a_r}_{b_1,\cdots,b_s}$ be a tensor of type $r+s$.
Then at a given point $x\in \cM$, we define its magnitude to be
\begin{align}\label{eq1:8july12}
|v| = (v^{a_1,\cdots,b_s}v_{a_1,\cdots,b_s})^{\frac12},
\end{align}
where the indices of $v$ are raised and lowered with respect to $g_{ab}$.
We then define the Banach space of $k$-differentiable functions  
$C^{k}(\cM\times \mathbb{R})$ with norm $\|\cdot\|_k$
to be those functions $u$ satisfying
$$
\|u\|_k = \sum_{j=0}^k \sup_{x\in\cM}|D^ju| < \infty,
$$
where $D$ is the covariant derivative associated with $g_{ab}$.
Similarly, we define the space $C^k(\cT^r_s \cM)$ of $k$-times
differentiable $(r,s)$ tensor fields to be those tensors $v$ satisfying $\|v\|_k < \infty$.

Given two points $x,y \in \cM$, we define $d(x,y)$ to be the geodesic distance between them.
Let $\alpha \in (0,1)$.
Then we may define the $C^{0,\alpha}$ H\"{o}lder seminorm for a scalar-valued function $u$ to be
$$
[u]_{0,\alpha} = \sup_{x \ne y} \frac{|u(x)-u(y)|}{(d(x,y))^{\alpha}}.
$$
Using parallel transport, this definition can be extended to $(r,s)$-tensors
$v$ to obtain the $C^{k,\alpha}$ seminorm $[u]_{k,\alpha}$ \cite{Au82}.
This leads us to the following definition of the $C^{k,\alpha}(\cM\times \mathbb{R})$ H\"{o}lder norm
$$
\|u\|_{k,\alpha} = \|u\|_k + [u]_{k,\alpha}
$$
for scalar-valued functions, and we may define the $C^{k,\alpha}(\cT^r_s \cM)$ H\"{o}lder
norm for $(r,s)$ tensors in a similar fashion.

Finally, we will also make use of the Sobolev spaces $ W^{k,p}(\cM\times \mathbb{R})$ 
and \\
$W^{k,p}(\cT^r_s \cM)$ where we assume
$k \in \mathbb{N}$ and $p \ge 1$.
If $dV_g$ denotes
the volume form associated with $g_{ab}$, then the 
$L^p$ norm of an $(r,s)$ tensor is defined to be 
\begin{align}\label{eq3:8july12}
\|v\|_p = \left( \int_{\cM} |v|^p dV_g\right)^{\frac1p}.
\end{align}  
We can then define the Banach space $W^{k,p}(\cM\times \mathbb{R})$ (resp.\ $W^{k,p}(\cT^r_s\cM)$) to be those
functions (resp.\ $(r,s)$ tensors) $v$ satisfying
$$
\|v\|_{k,p} = \left( \sum_{j=0}^k\|D^jv\|^p_p \right)^{\frac1p} < \infty.
$$

The above norms are independent of the
background metric chosen.
Indeed, given any two metrics
$g_{ab}$ and $\hat{g}_{ab}$, one can show that the norms induced
by the two metrics are equivalent.
For example, if $D$ and $\hat{D}$ are
the derivatives induced by $g_{ab}$ and $\hat{g}_{ab}$ respectively,
then there exist constants $C_1$ and $C_2$ such that
$$
C_1 \|u\|_{k,\hat{g}} \le \|u\|_{k,g} \le C_2 \|u\|_{k,\hat{g}},
$$
where $\|\cdot\|_{k,g}$ denotes the $C^k(\cM)$ norm with respect to $g$.
This holds for
the $W^{k,p}$ and $C^{k,\alpha}$ norms as well.
We also note that the above norms
are related through the Sobolev embedding theorem.
In particular, the spaces $C^{k,\alpha}$ and $W^{l,p}$ are related in the sense that
if $n$ is the dimension of $\cM$ and $u \in W^{l,p}$ and
$$
k+\alpha < l-\frac np, 
$$
then $u \in C^{k,\alpha}$.
See \cite{Au82,Au98,EH96,RP68} for a complete
discussion of the Sobolev embedding Theorem, Banach spaces on manifolds, and the above norms.
 
\subsection{Adjoints and Projection Operators}

Solutions to the coupled system \eqref{eq3:3july12} satisfy
\begin{align}\label{eq4:10july12}
F(x,{\bw}) = 0,
\end{align}
where $F:X\times Y\to Z$ is a nonlinear operator between Banach spaces.
This allows us to use basic tools from functional analysis to analyze our 
problem.
In particular, we will repeatedly need to consider the linearization
$D_xF(x,{\bw})$, its adjoint, and projections onto subspaces determined
by these operators.
Later on in the section when we introduce the
Liapunov-Schmidt reduction, we will use the kernel of the linearization
$D_xF(x_0,\bw_0)$ at a point $(x_0,\bw_0)$,
the kernel of the adjoint, and projection operators
onto these subspaces, to decompose $X$ and $Y$ in a manner that
will greatly simplify our analysis.
Here we briefly discuss the adjoint
and projection operators.
See \cite{EZ86} for a more complete discussion of these topics
and see Appendix~\ref{sec1:10july12} for a discussion of Fr\'{e}chet derivatives.

\subsubsection{The Adjoint and Properties}

Suppose that $\cH$ is a Hilbert space with inner product $\langle \cdot, \cdot \rangle$.
Then if $A:\cH \to \cH$ is a linear operator, the Riesz Representation Theorem
implies that there exists a unique operator $A^*$ that satisfies
\begin{align}\label{eq1:10july12}
\langle Ax,y \rangle = \langle x,A^*y \rangle \quad \text{for all $x,y \in \cH$}.
\end{align}
\noindent
If $R(A)$ denotes the range of $A$ and $\text{ker}(A)$ denotes the kernel, then the operator $A^*$ satisfies the following properties:
\begin{align}\label{eq2:10july12}
&{\bf 1)} \quad \text{ker}( A^*) = R(A)^{\perp}\\
&{\bf 2)} \quad (\text{ker}(A^*))^{\perp} = \overline{R(A)}.
\end{align}

\subsubsection{Projection Operators and Fredholm Operators}

Now assume that $X \subset \cH$ is a Banach space contained in a Hilbert space $\cH$.
Given a subspace $V \subset X$, the projection
$P$ onto $V$  is a bounded linear operator ${P:X \to V}$  that satisfies $P^2 = P$.
In particular, if $V$ is a finite-dimensional
subspace spanned by the orthonormal basis $\hat{v}_1,\cdots, \hat{v}_n$, then we can easily construct the projection
onto $V$ by the formula
\begin{align}\label{eq1:12july12}
Pu: \sum_{i=1}^n \langle u,\hat{v}_i\rangle \hat{v}_i,
\end{align} 
where $u \in X$ and $\langle \cdot, \cdot \rangle$ is the inner product on $\cH$.
Note that $P$ is just the
normal projection operator from $\cH$ to $V$ restricted to $X$.

We end the section by introducing one more definition that will be important in the following section.
A {\bf Fredholm operator} is a bounded linear operator $A:X\to Y$ where $X$ and $Y$ are Banach spaces such that
$\text{dim ker}(A)$ and $\text{dim ker}(A^*)$ are finite-dimensional and $R(A)$ is closed.
Given a nonlinear operator $F:U \to Y$ where $U\subset X$, we say that $F$ is a {\bf nonlinear Fredholm operator}
if it is Fr\'{e}chet differentiable on $U$ and $D_xF(x)$ is a Fredholm operator.

Notice that if $A$ is a Fredholm operator, then $\text{ker}(A^*)^{\perp} = R(A)$ and furthermore, the fact that $\text{ker}(A)$
and $\text{ker}(A^*)$ are finite dimensional allows one to define projection operators $P$ and $Q$ onto 
$\text{ker}(A)$ and $\text{ker}(A^*)$ to decompose $X$ and $Y$.
As we will see, these properties make Fredholm operators
ideal candidates for bifurcation analysis.

\subsection{Elements of Bifurcation Theory}\label{LSReduc}

We now present some basic concepts from bifurcation theory that will be essential in obtaining our non-uniqueness
results.
In particular, we give a formal definition of a bifurcation point and then present the Liapunov-Schmidt reduction.
This reduction allows one to reduce a nonlinear problem between infinite-dimensional Banach spaces
to a finite-dimensional or even scalar-valued problem.
Therefore it greatly simplifies the analysis
and will serve as a basic tool for us going forward.
The following treatment is taken from \cite{HK04} and \cite{CH82}.

Suppose that $F: U\times V \to Z$ is a mapping with open sets $U \subset X,V \subset \Lambda$,
where $X$ and $Z$ are Banach spaces and $\Lambda = \mathbb{R}$.
We let $x\in X$ and $\lambda \in \La$.
Additionally assume that $F(x,\lambda)$ is Fr\'{e}chet differentiable with respect to $x$ and $\lambda$ on $U\times V$.
We are interested in solutions to the nonlinear problem
\begin{align}\label{eq3:2july12}
F(x,\la) = 0.
\end{align}
A solution of \eqref{eq3:2july12} is a point $(x,\la) \in X\times \La$ such
that \eqref{eq3:2july12} is satisfied.

\begin{definition}
Suppose that $(x_0,\la_0)$ is a solution to \eqref{eq3:2july12}.
We say that $\la_0$ is a {\bf bifurcation point} if 
for any neighborhood $U$ of $(x_0,\la_0)$ there exists a $\la\in \La$ and $x_1, x_2 \in X$, $x_1 \ne x_2$
such that $(x_1,\la), (x_2,\la) \in U$ and $(x_1,\la)$ and $(x_2,\la)$ are both solutions to \eqref{eq3:2july12}.
\end{definition}

Given a solution $(x_0,\la_0)$ to \eqref{eq3:2july12}, we are interested in analyzing solutions to \eqref{eq3:2july12}
in a neighborhood of $(x_0,\la_0)$ to determine whether it is a bifurcation point.
One of the most useful tools for this is the Implicit Function
Theorem~\ref{thm1:2july12}.
This theorem asserts that if $D_xF(x_0,\la_0)$ is invertible, then there exists a 
neighborhood $U_1\times V_1 \subset U\times V$ and a continuous function $f: V_1\to U_1$ such that all solutions
to \eqref{eq3:2july12} in $U_1\times V_1$ are of the form $(f(\la),\la)$.
Therefore in order for a bifurcation to occur at $(x_0,\la)$,
it follows that $D_xF(x_0,\la_0)$ must not be invertible.

\subsubsection{Liapunov-Schmidt Reduction}

The following discussion is taken from \cite{HK04}.
Let $X, \Lambda$ and $Z$ be Banach spaces and assume that $U\subset X$, $V\subset \La$.
For $\la = \la_0$, we require that the mapping 
${F:U\times V \to Z}$ be a 
nonlinear Fredholm operator with respect to $x$; i.e.\ the linearization $D_xF(\cdot,\la_0)$ of $F(\cdot,\la_0):U \to Z$ 
is a Fredholm operator.
Assume that $F$ also satisfies the following assumptions:
\begin{align}\label{eq7:2july12}
&F(x_0,\la_0) = 0 \quad \text{for some $(x_0,\la_0) \in U \times V$},\\
&\text{dim ker}(D_xF(x_0,\la_0)) = \text{dim ker}(D_xF(x_0,\la_0)^*) = 1. \nonumber
\end{align}
Given that $D_xF(x_0,\lambda_0)$ has a one-dimensional kernel, 
there exists a projection operator $P:X \to X_1 = \text{ker}(D_xF(x_0,\la_0))$.
Similarly, one has the projection operator
${Q:Y\to Y_2 = \text{ker}(D_xF(x_0,\la_0)^*)}$.
This allows us to decompose $X = X_1 \oplus X_2$ and $Y = Y_1 \oplus Y_2$ where 
$Y_1 = R(D_XF(x_0,\lambda_0))$.
We will refer to the decomposition $X_1\oplus X_2$ and $Y_1\oplus Y_2$ induced by
$D_xF(x_0,\la_0)$ as the {\bf Liapunov decomposition}, and we see that  
$F(x,\lambda) = 0$ if and only if the following two equations are satisfied
\begin{align}\label{eq1:27june12}
&QF(x,\lambda) = 0,\\
&(I-Q)F(x,\lambda) = 0. \nonumber
\end{align}

For any $x\in X$, we can write $x = v+w$, where $v= Px$ and $w = (I-P)x$.
Define $G: U_1\times W_1 \times V_1 \to Y_1$ by
\begin{align}\label{eq5:2july12}
&G(v,w,\la) = (I-Q)F(v+w,\la), \quad \text{where}\\
&U_1\subset X_1,\hspace{3mm} W_1 \subset X_2, \hspace{3mm}   V_1 \subset \mathbb{R} \quad \text{and} \nonumber \\
&v_0 = Px_0 \in U_1, \quad w_0 = (I-P)x_0 \in W_1, \nonumber
\end{align}
and $U_1, W_1$ are neighborhoods such that $U_1 + W_1 \subset U \subset X$.

Then the definition of $G(v,w,\la)$ implies that $G(v_0,w_0,\la_0) = 0$ and our choice of function spaces ensures that
$$D_wG(v_0,w_0,\la_0) =(I- Q)D_xF(x_0,\la_0):X_2 \to Y_1,$$
is bijective.
The Implicit Function Theorem then implies that there exist
neighborhoods $U_2\subset U_1, W_2\subset W_1$ and $V_2\subset V_1$ and
a continuous function  
\begin{align}\label{eq9:6july12}
&\psi:U_2\times V_2 \to W_2 \quad \text{such that all solutions to $G(v,w,\la) = 0$}\\
&\text{ in $U_2\times W_2\times V_2$ $\quad$ are of the form $G(v,\psi(v,\la), \la ) = 0.$}\nonumber 
\end{align}
Insertion of the function $\psi(v,\la)$ into the second equation of \eqref{eq1:27june12} yields
a finite-dimensional problem
\begin{align}\label{eq6:2july12}
\Phi(v,\la) = QF(v+\psi(v,\la),\la) = 0.
\end{align}
We observe that finding solutions $(v,\la)$ to \eqref{eq6:2july12} is equivalent to finding
solutions to $F(x,\la) = 0$ in a neighborhood of $(x_0,\la_0)$.
We will refer to the finite-dimensional problem \eqref{eq6:2july12}
as the {\bf Liapunov-Schmidt reduction} of \eqref{eq3:2july12}.

Given that $\text{ker}(D_xF(x_0,\lambda_0))$ is spanned by $\hat{v}_0$, then we can write $v = s\hat{v}_0 +v_0$.
Substituting this into \eqref{eq6:2july12} we obtain
\begin{align}\label{eq2:27jun12}
\Phi(s,\la) = QF(s\hat{v}_0+v_0+\psi(s\hat{v}_0+v_0,\la),\lambda) = 0.
\end{align}  
Using the reduction \eqref{eq2:27jun12} and another application of the Implicit Function Theorem, 
one obtains the following theorem taken from \cite{HK04}, which allows us to determine
a unique solution curve through the point $(x_0,\la_0)$.
We also include the proof for completeness.

\begin{theorem}\label{thm1:29apr12}
Assume $F:U\times V \to Z$ is continuously differentiable on ${U\times V \subset X \times \mathbb{R}}$ and
that assumptions \eqref{eq7:2july12} hold.
Additionally we assume that 
\begin{align}\label{eq3:29apr12}
D_{\lambda}F(x_0,\lambda_0) \notin R(D_xF(x_0,\lambda_0)).
\end{align}
Then there is a continuously differentiable curve through $(x_0,\lambda_0)$; that is, there exists
\begin{align}\label{eq4:29apr12}
\{(x(s),\lambda(s))~|~s \in (-\delta, \delta),~ (x(0),\lambda(0)) = (x_0,\lambda_0)\},
\end{align}
such that
\begin{align}\label{eq5:29apr12}
F(x(s), \lambda(s)) = 0 \quad \text{for $s\in (-\delta, \delta)$},
\end{align}
and all solutions of $F(x,\lambda) = 0$ in a neighborhood of $(x_0, \lambda_0)$ belong to the curve
\eqref{eq4:29apr12}.
\end{theorem}
\begin{proof}
Let $x_0 = v_0 + w_0 = v_0 +\psi(v_0,\la_0)$.
Differentiating \eqref{eq6:2july12} with respect to $\la$ we obtain
\begin{align}\label{eq1:6july12}
&D_{\la}\Phi(v_0,\la_0) = \\
&QD_xF(x_0,\la_0)D_{\la}\psi(v_0,\la_0) + QD_{\la}F(x_0,\la_0) =  QD_{\la}F(x_0,\la_0) \ne 0,\nonumber
\end{align}
where \eqref{eq1:6july12} is nonzero due to the extra assumption \eqref{eq3:29apr12}.
The above expression simplifies due to the fact that that 
$$D_xF(x_0,\la_0)D_{\la}\psi(v_0,\la_0) \in R(D_xF(x_0,\la_0)),$$
and $Q$ is the projection onto $\text{ker}(D_XF(x_0,\la_0)^*)$.

The fact that $D_{\la}\Phi(v_0,\la_0)  \ne 0$ and that $X_1, Y_2$ and $\mathbb{R}$ are one-dimensional implies that
we may apply the Implicit Function Theorem to $\Phi(v,\la)$ to conclude that there exists a continuously differentiable
$\gamma:U_2 \to V_2 \subset \mathbb{R}$ such that 
\begin{align}\label{eq2:6july12}
\gamma(v_0) = \la_0 \quad \text{and} \quad \Phi(v,\gamma(v)) = 0 \quad \text{for all} \quad v \in U_2 \subset X_1.
\end{align}
Therefore our reduced equation \eqref{eq6:2july12} becomes
\begin{align}\label{eq3:6july12}
\Phi(v,\gamma(v)) = QF(v+\psi(v,\gamma(v)),\gamma(v)) = 0,
\end{align}
where solutions to \eqref{eq3:6july12} are of the form 
\begin{align}\label{eq4:6july12}
x(v) = v+\psi(v,\gamma(v)) \quad \text{and} \quad \la(v) = \gamma(v).
\end{align}
By writing $v = s\hat{v}_0 + v_0$ as in \eqref{eq2:27jun12} and inserting this into \eqref{eq4:6july12}, we obtain our solution curve
\begin{align}
&x(s) = v_0 +s\hat{v}_0 + \psi(v_0+s\hat{v}_0,\gamma(v_0+s\hat{v}_0)),  \label{eq11:6july12} \\
&\la(s) = \gamma(v_0+s\hat{v}_0). \label{eq12:6july12}
\end{align}
\end{proof}

Now we compile some useful properties of the maps $\Phi(v,\la)$, $\psi(v,\la)$ and $\gamma(v)$
defined in the \eqref{eq6:2july12}, \eqref{eq2:6july12} and \eqref{eq9:6july12}.
These results, along with their proofs, are taken from \cite{HK04}.

\begin{proposition}\label{prop1:6july12}
Let the assumptions of Theorem~\ref{thm1:29apr12} hold and let
the operators $\Phi(v,\la)$, $\psi(v,\la)$ and $\gamma(v)$ be defined as in 
\eqref{eq6:2july12}, \eqref{eq2:6july12} and \eqref{eq9:6july12} and let $\la_0$ and ${x_0 = v_0+w_0}$
be as in the previous discussion.
Then
\begin{align}\label{eq10:6july12}
D_v\Phi(v_0,\la_0) = 0, \quad D_v\psi(v_0,\la_0) = 0, \quad \text{and} \quad D_v\gamma(v_0) = 0,
\end{align}
and each of these operators has the same order of differentiability as $F(x,\la)$.
\end{proposition}

\begin{proof}
The fact that $\Phi(v,\la)$, $\psi(v,\la)$ and $\gamma(v)$ all have the same order
of differentiability as $F(x,\la)$ follows from the definition of $\Phi(v,\la)$ and the
Implicit Function Theorem~\ref{thm1:2july12}.
By differentiating
$(I-Q)F(v+\psi(v,\la),\la) = 0$
with respect to $v$ we obtain
\begin{align}
(I-Q)D_xF(v+\psi(v,\la),\la)(I_{X_1}+D_v\psi(v,\la)) = 0,
\end{align}
where $I_{X_1}$ denotes the identity on $X_1 = \text{ker}(D_xF(x_0,\la_0))$.
By evaluating at $(v_0,\la_0)$, where $x_0 = v_0+w_0$, we obtain
\begin{align}\label{eq7:5july12}
(I-Q)D_xF(x_0, \la_0)D_v\psi(v_0,\la_0) = 0.
\end{align}
Given that $D_v\psi(v_0,\la_0)$ maps onto $X_2$ and $(I-Q)D_XF(x_0,\la_0)$
is an invertible operator from $X_2$ to $Y_1$, we have that $D_v\psi(x_0,\la_0)= 0$.

Then if we differentiate $\Phi(v,\la) = QF(v+\psi(v,\la),\la) = 0$ with respect to $v$ and evaluate at $(v_0,\la_0)$,
we obtain
\begin{align}\label{eq5:6july12}
D_v\Phi(v_0,\la_0) = QD_xF(x_0,\la_0)I_{X_1} = 0.
\end{align}
By differentiating \eqref{eq3:6july12} with respect to $v$ and utilizing \eqref{eq5:6july12},
we have
$$
D_{\la}\Phi(v_0,\la_0)D_v\gamma(v_0) = 0.
$$
The assumption that $D_{\la}\Phi(v_0,\la_0) \ne 0$ implies that 
\begin{align}\label{eq6:6july12}
D_v\gamma(v_0) = 0.
\end{align}
\end{proof}

Once we've obtained a unique solution curve $(x(s),\la(s))$ through $(x_0,\la_0)$, we
analyze $\ddot{\la}(0)$ (where $\dot{} = \frac{d}{ds}$) to determine additional information
about the solution curve.
In particular, we can determine whether
or not a {\bf saddle node bifurcation} or fold occurs at $(x_0,\la_0)$.
This type of bifurcation
occurs when the solution curve $\{x(s),\la(s)\}$ has a turning point at $(x_0,\la_0)$.
The next proposition, taken
from \cite{HK04}, provides us with a method to determine information about $\ddot{\la}(0)$.

\begin{proposition}\label{prop1:1july12}
Let the assumptions of Theorem~\ref{thm1:29apr12} be in effect.
Additionally assume that 
$ker(D_XF(x_0,\la_0))$ is spanned by $\hat{v}_0$.
Then 
\begin{align}\label{eq1:1july12}
&\left.\frac{d}{ds}F(x(s),\la(s))\right|_{s=0}= \\
& D_xF(x_0,\la_0)\dot{x}(0)+D_{\la}F(x_0,\la_0)\dot{\la}(0) =D_xF(x_0,\la_0)\hat{v}_0= 0 \nonumber\\
&\left.\frac{d^2}{ds^2}F(x(s),\la(s))\right|_{s=0}= \label{eq1:17july12} \\
&D^2_{xx}F(x_0,\la_0)[\hat{v}_0,\hat{v}_0]+D_xF(x_0,\la_0)\ddot{x}(0)+D_{\la}F(x_0,\la_0)\ddot{\la}(0) = 0.\nonumber
\end{align}
In particular, an application of the projection operator
$Q$ defined in \eqref{eq1:27june12} to \eqref{eq1:17july12} yields
\begin{align}\label{eq3:5oct12}
QD^2_{xx}F(x_0,\la_0)[\hat{v}_0,\hat{v}_0]+QD_{\la}F(x_0,\la_0)\ddot{\la}(0) = 0.
\end{align}
This implies that if $D_{\la}F(x_0,\la_0)\notin R(D_xF(x_0,\la_0))$ and 
$$D^2_{xx}F(x_0,\la_0)[\hat{v}_0,\hat{v}_0] \notin R(D_xF(x_0,\la_0)),$$
then $\ddot{\la}(0) \ne 0$.
\end{proposition}
\begin{proof}
Let $\{x(s),\la(s)\}$ be the solution curves for $F(x,\la) = 0 $ defined by \eqref{eq11:6july12} and \eqref{eq12:6july12}.
Differentiating
these curves we obtain
\begin{align}
&\left.\frac{d}{ds}x(s)\right|_{s=0} = \hat{v}_0 + D_v\psi(v_0,\la_0)\hat{v}_0+D_{\la}\psi(v_0,\la_0)D_v\gamma(v_0)\hat{v}_0 = \hat{v}_0, \label{eq7:6july12} \\
&\left.\frac{d}{ds}\la(s)\right|_{s=0} =D_v\gamma(v_0)\hat{v}_0 = 0, \label{eq13:6july12}
\end{align}
where the above expressions simplify as a result of Proposition~\ref{prop1:6july12}.
Differentiating the expression $F(x(s),\la(s)) = 0$ twice and again using Proposition~\ref{prop1:6july12} to simplify, we obtain
\begin{align}\label{eq8:6july12}
&\left.\frac{d^2}{ds^2}F(x(s),\la(s))\right|_{s=0} =\\
& D^2_{xx}F(x_0,\la_0)[\hat{v}_0,\hat{v}_0]+D_xF(x_0,\la_0)\ddot{x}(0)+D_{\la}F(x_0,\la_0)\ddot{\la}(0) = 0,\nonumber  
\end{align}
where 
$$
\ddot{\la}(0) = D^2_{vv}\gamma(v_0)[\hat{v}_0,\hat{v}_0] \quad \text{and} \quad \ddot{x}(0) = D^2_{vv}\psi(v_0,\la_0)[\hat{v}_0,\hat{v}_0],
$$
by differentiating \eqref{eq7:6july12} and \eqref{eq13:6july12} once more with respect to $s$.
Applying the projection operator $Q$ to \eqref{eq8:6july12} yields \eqref{eq3:5oct12}.
Then the assumptions that
$D_{\la}F(x_0,\la_0) \notin R(D_xF(x_0,\la_0))$ and $D^2_{xx}F(x_0,\la_0)[\hat{v}_0,\hat{v}_0] \notin R(D_xF(x_0,\la_0))$
imply that $\ddot{\la}(0) \ne 0$.
\end{proof}

The significance of Proposition~\ref{prop1:1july12} is that it gives explicit conditions that allow us to determine whether or not $\ddot{\la}(0)$ is nonzero.
Heuristically, the fact that $\ddot{\la}(0) \ne 0$ means that $\la(s)$ has a turning point at $s=0$.
This means
that the graph of $\{x(s),\la(s)\}$ looks like a parabola and that a saddle node
bifurcation occurs at $s=0$ (cf.~\cite{HK04}).
If we assume that $F(x,\la)$ is at least $3$-times differentiable we may expand the operators ${\psi(v_0+s\hat{v}_0,\gamma(v_0+s\hat{v}_0))}$ and ${\gamma(v_0+s\hat{v}_0)}$ about $s=0$ as a second order Taylor series and use \eqref{eq11:6july12} and \eqref{eq12:6july12} to obtain second order representations of our solutions $\{x(s),\la(s)\}$.
This is the solution approach we take to prove non-uniqueness in both the CMC and non-CMC cases.

\section{Main Results}\label{sec2:14july12}

The main results of this article pertain to the following one parameter family 
of problems
\begin{align}\label{eq2:30apr12}
-\Delta \phi + &a_R\phi + \lambda^2a_{\tau}\phi^5-a_{{\bf w}}\phi^{-7}-2\pi\rho e^{-\lambda}\phi^{5}=0,\\
&\mathbb{L}{\bf w} + \lambda b_{\tau}^a\phi^6=0. \nonumber
\end{align}
Here we assume that $g_{ab}$ is a given SPD metric with no conformal killing fields that has constant, positive scalar curvature.
The expressions $D_a$ and $\Delta$ denote the derivative and the Laplace-Beltrami operator associated with $g_{ab}$ and
$$\mathbb{L}\bw = -D_b(\mathcal{L}\bw)^{ab},$$
\noindent
denotes the divergence of the
conformal killing operator associated with $g_{ab}$.
Finally, we define
\begin{align}\label{eq5:27jun12}
&a_R = \frac{1}{8}R,  & a_{\tau} = \frac{1}{12}\tau^2, \\
&a_{{\bf w}} = \frac{1}{8}(\sigma+\mathcal{L}{\bf w})_{ab}(\sigma +\mathcal{L}{\bf w})^{ab},  & b_{\tau} = \frac{2}{3}D^a \tau. \nonumber
\end{align}
In general, we assume that $\tau \in C^{1,\alpha}(\cM)$, however when we prove our CMC results we will
additionally require that $\tau$ be constant.
For the remainder of this paper we assume that $R$ is a positive constant and that $|\sigma| = (\sigma_{ab}\sigma^{ab})^{\frac12}$
is also a nonzero constant.
Notice that \eqref{eq2:30apr12} has the form of \eqref{eq3:3july12} with 
initial data depending on $\la$ where
$$
\tau_{\la} = \la\tau, \quad \rho_{\la} = e^{-\la}\rho \quad \text{and} \quad \bj_{\la} = 0.
$$ 

We show that in both the CMC and non-CMC cases that solutions to \eqref{eq5:27jun12} are non-unique.
Our method
for doing this is to apply the bifurcation theory outlined in Section~\ref{LSReduc}.
The first step in doing this is
to formulate \eqref{eq2:30apr12} in a way that allows us to utilize the framework outlined in Section~\ref{LSReduc}.

\subsection{Problem Setup}

We now formulate \eqref{eq2:30apr12} so that we can apply the Liapunov-Schmidt reduction.
Define $F((\phi, {\bf w}), \lambda)$ by
\begin{align}\label{eq6:27jun12}
F((\phi,{\bf w}),\lambda) = \left[ \begin{array}{c} -\Delta \phi + a_R\phi + \lambda^2 a_{\tau}\phi^5-a_{{\bf w}}\phi^{-7}-2\pi\rho e^{-\lambda}\phi^{5}\\
\mathbb{L}{\bf w} + \lambda b_{\tau}^a\phi^6 
 \end{array} \right],
\end{align}
and in the event that $\tau$ is constant, define
\begin{align}\label{eq10:29jun12}
G(\phi,\la) = -\Delta\phi+ a_R\phi + \lambda^2 a_{\tau}\phi^5-\frac18\sigma^2\phi^{-7}-2\pi\rho e^{-\lambda}\phi^{5}.
\end{align}
If $F((\phi,{\bf w}),\lambda) = 0$ (resp.\ $G(\phi,\la) = 0$) for a given $\lambda$, then $((\phi,{\bf w}),\la)$ (resp.\ $(\phi,\la)$) solves Eq.\ \eqref{eq2:30apr12} (resp.\ Eq.\ \eqref{eq10:29jun12}).

We view \eqref{eq6:27jun12} and \eqref{eq10:29jun12} as nonlinear operators between the Banach spaces 
$$
F((\phi,\bw),\la):C^{k,\alpha}(\cM)\oplus C^{k,\alpha}(\cT \cM)\times \mathbb{R} \to C^{k-2,\alpha}(\cM)\oplus C^{k-2,\alpha}(\cT \cM),
$$
$$
G(\phi,\la):C^{k,\alpha}(\cM)\times \mathbb{R} \to C^{k-2,\alpha}(\cM).
$$
where $k \ge 2$.
For $\phi \ne 0$ and $X=(\phi,\bw)$, the first order Fr\'{e}chet derivatives $D_{\phi}G(\phi,\la),\\
D_{\la}G(\phi,\la)$, $D_XF((\phi,\bw),\la)$ and $D_{\la}F((\phi,\bw),\la)$
all exist.
In fact, both $F$ and $G$ are $k$-differentiable for any $k\in \mathbb{N}$ provided that $\phi \ne 0$.
See the Appendix~\ref{sec1:10july12} for more information regarding Fr\'{e}chet derivatives.

Now we are ready to state the main results of this paper.
The first two results state that there is a critical density $\rho=\rho_c$
such that there exists a constant $\phi_c$ where the linearizations $D_{\phi}G(\phi_c,0)$ and $D_XF((\phi_c,{\bf 0}),0)$ have a kernel of dimension one.
This provides the basis for our final two main results where we determine explicit solution curves $\{\phi(s), \la(s)\}$ and $\{(\phi(s),\bw(s)),\la(s)\}$ to obtain our non-uniqueness results.

\subsection{Existence of $\rho_c$ such that $\text{ dim ker}(D_XF((\phi_c,{\bf 0},0)) = 1$}\label{sec1:6july12}

The two results in this section pertain to the existence of a critical energy density $\rho=\rho_c$ at which the linearizations of the operators
$F$ and $G$ develop a one-dimensional kernel.
These results allow us to apply the Liapunov-Schmidt reduction outlined in Section~\ref{LSReduc} to analyze solutions in a neighborhood of $((\phi_c,{\bf 0}),0)$ and $(\phi_c,0)$.
We present the theorems here without proof and postpone them until Section~\ref{sec1:29jun12}.

\begin{theorem}[CMC]\label{thm4:29jun12}
Let $D_{\phi}G(\phi,\la)$ denote the Fr\'{e}chet derivative of \eqref{eq10:29jun12} with respect to $\phi$.
Then
there exists a critical value of $\rho= \rho_c~$ and a constant $\phi_c$ such that when $\rho= \rho_c$, Eq.\ \eqref{eq10:29jun12} has a solution if and only if $\la \ge 0$.
Furthermore, $\text{dim ker}(D_{\phi}G(\phi_c,0))) = 1$ and it is spanned by the constant function $\phi =1$.
Moreover, we can determine the explicit values of $\rho_c$ and $\phi_c$, which are
\begin{align}\label{eq1:29jun12}
\rho_c = \frac{R^{\frac32}}{24\sqrt{3}\pi |\sigma|} \quad \text{and} \quad \phi_c = \left(\frac{R}{24\pi \rho}\right)^{\frac14}.
\end{align}
\end{theorem}
\begin{proof}
We present the proof in Section~\ref{sec1:29jun12}.
\end{proof}

\begin{theorem}[non-CMC]\label{thm2:29jun12}
Let $D_XF((\phi,\bw),\la)$ denote the Fr\'{e}chet derivative of Eq.\ \eqref{eq6:27jun12} with respect to $X=(\phi,\bw)$ and
let $\rho_c$ and $\phi_c$ be as in Theorem~\ref{thm4:29jun12}.
Then when $\rho= \rho_c$, $\text{dim ker}(D_XF((\phi_c,{\bf 0}),0))) = 1$ and it is spanned
by the constant vector $\tiny{\left[\begin{array}{c} 1 \\ 0 \end{array} \right]}$.
\end{theorem}
\begin{proof}
We present the proof in Section~\ref{sec1:29jun12}.
\end{proof}

\subsection{Non-unique Solutions to $F((\phi,\bw),\la)=0$ when $\rho=\rho_c$}

The two Theorems in this section pertain to the non-uniqueness of solutions to the nonlinear problems
\eqref{eq6:27jun12} and \eqref{eq10:29jun12}.
Theorem~\ref{thm3:29jun12} provides the explicit form of solutions to \eqref{eq10:29jun12} in a neighborhood of the point $(\phi_c,0)$ in the CMC case.
The form of this solution curve implies that a saddle node bifurcation occurs at
$(\phi_c,0)$ and that solutions are non-unique in a neighborhood of this point.
Theorem~\ref{thm1:29jun12} provides analogous results in the non-CMC case for the point $((\phi_c,{\bf 0}),0)$.

\begin{theorem}[CMC]\label{thm3:29jun12}
Suppose that $\tau$ is constant.
Then \eqref{eq6:27jun12} reduces to the scalar problem
\begin{align}\label{eq6:29jun12}
-\Delta\phi + a_R\phi +(\la^2a_{\tau}-2\pi\rho e^{-\la})\phi^5 - \frac18\sigma^2\phi^{-7} = 0.
\end{align}
When $\rho = \rho_c$, with $\rho_c$ as in Theorem~\ref{thm2:29jun12}, 
then there exists a neighborhood
of $(\phi_c,0)$ such that all solutions to \eqref{eq6:29jun12} 
in this neighborhood lie on
a smooth solution curve $\{\phi(s),\la(s)\}$ that has the form
\begin{align}\label{eq17:6july12}
&\phi(s) = \phi_c + s + O(s^2),\\
&\la(s) = \frac12\ddot{\la}(0)s^2 + O(s^3), \quad (\ddot{\la}(0) \ne 0).
\end{align}
In particular, there exists a $\delta>0$ such that for all $0 < \la < \delta$ there exist at least two
distinct solutions $\phi_{1,\la} \ne \phi_{2,\la}$ to \eqref{eq6:29jun12}.
\end{theorem}
\begin{proof}
We postpone the proof until Section~\ref{sec4:14july12}.
\end{proof}

\begin{theorem}[non-CMC]\label{thm1:29jun12}
Suppose $\tau \in C^{1,\alpha}(\cM)$ is non-constant and let
${F((\phi,\bw),\la)}$ be defined as in \eqref{eq6:27jun12}.
Then if $\rho_c$ and $\phi_c$ are defined as in Theorem~\ref{thm4:29jun12} and $\rho = \rho_c$, there exists a neighborhood of $((\phi_c,\bw),0)$ such that
all solutions to ${F((\phi,\bw),\la)=0}$ in this neighborhood lie on a smooth curve of the form
\begin{align}\label{eq14:6july12}
&\phi(s) = \phi_c + s + \frac12 \ddot{\la}(0) u(x) s^2 + O(s^3),\\
&\bw(s) = \frac12 \ddot{\la}(0) \bv(x) s^2 + O(s^3),  \nonumber\\
&\la(s) = \frac12\ddot{\la}(0)s^2+O(s^3), \quad (\ddot{\la}(0) \ne 0),\nonumber
\end{align}
where $u(x) \in C^{2,\alpha}(\cM)$, $\bv(x) \in C^{2,\alpha}(\cT\cM)$ and $\bv(x) \ne {\bf 0}$.
In particular, there exists a $\delta>0$ such that for all $0<\la < \delta$ there exist elements $(\phi_{1,\la},\bw_{1,\la}),(\phi_{2,\la},\bw_{2,\la}) \in C^{2,\alpha}(\cM)\oplus C^{2,\alpha}(\cT \cM)$ such that 
$$F((\phi_{i,\la},\bw_{i,\la}),\la) = 0,~~  \text{for} ~~i \in \{1,2\},~~\text{and} ~~ 
(\phi_{1,\la},\bw_{1,\la}) \ne (\phi_{2,\la},\bw_{2,\la}).$$
\end{theorem}
\begin{proof}
We present the proof in Section~\ref{sec5:14july12}.
\end{proof}

\section{Some Key Technical Results}

\subsection{Existence of a Critical Value $\rho_c$}\label{sec3:14july12}

In this section we lay the foundation for proving Theorems \ref{thm4:29jun12} and ~\ref{thm2:29jun12}.
As in \cite{PY05}, we seek a critical density $\rho_c$ where our elliptic problem goes from having positive
solutions to having no positive solutions.
In particular, what we seek is a value $\rho_c$ such that when $\la = 0$,
then \eqref{eq6:27jun12} will have no solution for $\rho> \rho_c$ and will have a solution for $\rho \le \rho_c$.

When $\lambda = 0$, the assumption that $g_{ab}$ admits no conformal killing fields implies that
\begin{align}\label{eq3:29jun12}
 F((\phi,{\bf w}),0) = F((\phi,{\bf 0}),0) =\left[ \begin{array}{c} -\Delta \phi +a_R\phi - \frac{\sigma^2}{8}\phi^{-7}-2\pi\rho\phi^5=0 \\ {\bf w} = 0\end{array}\right].
\end{align}
Define 
\begin{align}\label{eq4:29jun12}
q(\chi) = a_R\chi -\frac18\sigma^2\chi^{-7} - 2\pi\rho_c\chi^5,
\end{align}
where $\rho_c$ is a constant to be determined.
The objective will be to determine
$\rho_c$ so that $q(\chi)$ has a single, positive, multiple root and then
use the maximum principle discussed in Appendix~\ref{maxp} to conclude that if $\rho > \rho_c$, then
\eqref{eq3:29jun12} will have no solution.
This leads us to the following proposition.

\begin{proposition}\label{prop1:16july12}
Let $q(\chi)$ be defined as in \eqref{eq4:29jun12}.
Then there exists constants $\rho_c>0$ and $\phi_c >0$ such that
$q(\chi) \le 0$ for all $\chi >0$ and the only positive root of $q(\chi)$ is $\phi_c$.
\end{proposition}
\begin{proof}
To determine $\rho_c$, we observe that because $a_R$ and $\sigma^2$ are constants, we simply need to analyze the roots of
\eqref{eq4:29jun12} as $\rho_c$ varies.
We seek $\rho_c$ such that $q(\chi)$ has a single, positive, multiple root.
We observe that
$q(\chi) = 0$ if and only if 
$$p(\chi) = a_R\chi^8-\frac18\sigma^2-2\pi\rho_c\chi^{12} = 0.$$  
Furthermore, it is clear that each pair of roots
$\{-\chi_0, \chi_0\}$ of the even polynomial $p(\chi)$ is in direct correspondence with each positive root of $p(\gamma)=a_R\gamma^2-\frac18\sigma^2-2\pi\rho_c\gamma^3$, where $\gamma = \chi^4$.
Therefore, we simply need to 
choose $\rho_c$ such that $p(\gamma)$ has a single positive root.
To accomplish this, we find the lone, local maximum of $p(\gamma)$
and require it to be a root of $p(\gamma)$.
We have that
$$
0 = p'(\gamma) = 2a_R\gamma-6\pi\rho_c\gamma^2 \Longrightarrow \gamma_c = \frac{a_R}{3\pi\rho_c} \hspace{3mm} \text{is a local max},
$$  
and 
\begin{align}
0 = p(\gamma_c) = a_R\left(\frac{a_R}{3\pi\rho_c}\right)^2- \frac18\sigma^2-2\pi\rho_c\left(\frac{a_R}{3\pi\rho_c}\right)\\
= \frac{a_R^3-\frac18\sigma^2(27\pi^2\rho_c^2)}{27\pi^2\rho_c^2}\Longrightarrow \rho_c = \frac{R^{\frac32}}{24\sqrt{3}|\sigma|\pi}. \nonumber
\end{align}
\end{proof}
The next result follows immediately from the previous analysis but will be useful going forward.
\newpage

\begin{corollary}\label{cor1:29jun12}
Define the constants
\begin{align}
\rho_c = \frac{R^{\frac32}}{24\sqrt{3}|\sigma|\pi} \quad \text{and} \quad \phi_c = \left(\frac{a_R}{3\pi \rho_c}\right)^{\frac14}.
\end{align}
Then if  
$$q(\chi) = a_R\chi -\frac18\sigma^2\chi^{-7}-2\pi\rho_c\chi^5,$$ 
it follows that $q(\phi_c) = q'(\phi_c) = 0$.
\end{corollary}
\begin{proof}
This follows immediately from the proof of Proposition~\ref{prop1:16july12} or by direct computation.
\end{proof}
Now we show that $\rho_c$ is a critical value of \eqref{eq3:29jun12}.
\begin{proposition}\label{prop1:29jun12}
Let $\rho(x) \in C(\cM)$.
Then the constant $\rho_c$ defined in Corollary~\ref{cor1:29jun12} has the
property that Eq.\ \eqref{eq3:29jun12} has a positive solution if
$~0< \rho \le \rho_c$ and has no positive solution if $\rho > \rho_c$.
\end{proposition}
\begin{proof}
Let $q(\chi)$ be defined as in Corollary~\ref{cor1:29jun12}.
If $\phi>0$ solves \eqref{eq3:29jun12}, then
\begin{align}\label{eq2:16july12}
\Delta \phi = a_R\phi-\frac18\sigma^2\phi^{-7}-2\pi\rho\phi^5 = f(x,\phi).
\end{align}
We observe that if $\rho > \rho_c$, then $\check{\rho} = \inf_{x\in\cM} \rho > \rho_c$ and for $\chi>0$,
\begin{align}\label{eq3:16july12}
f(x,\chi) = a_R\chi -\frac18\sigma^2\chi^{-7}-2\pi\rho\chi^5 \le a_R\chi -\frac18\sigma^2\chi^{-7}-2\pi\check{\rho}\chi^5 < q(\chi).
\end{align}
Therefore if $\rho >\rho_c$, \eqref{eq2:16july12} and \eqref{eq3:16july12} imply that any positive solution $\phi$ to \eqref{eq3:29jun12}
satisfies
$$
\Delta \phi = f(x,\phi) < q(\phi) \le 0.
$$
An application of the maximum principle \eqref{maxp} implies that if $\rho > \rho_c$, then \eqref{eq3:29jun12} has no solution.

To verify that
\eqref{eq3:29jun12} has a solution if $\rho \le \rho_c$, first observe that Corollary~\ref{cor1:29jun12} implies
that
\begin{align}\label{eq5:29jun12}
\phi_c = \left( \frac{a_R}{3\pi\rho_c}\right)^{\frac14} = \left( \frac{R}{24\pi\rho_c}\right)^{\frac14}, 
\end{align}
solves Eq.\ \eqref{eq3:29jun12} when $\rho = \rho_c$.
If $\rho < \rho_c$, the properties of $q(\chi)$ imply that the polynomial
$$q_1(\chi) = a_R\chi - \frac18\sigma^2\chi^{-7} -2\pi\hat{\rho}\chi^5, \quad \quad \hat{\rho} = \sup_{x\in\cM}\rho(x), $$ 
will have two positive roots $\chi_1 < \chi_2$.
Therefore, any $\phi_+$ satisfying $0<\chi_1 < \phi_+ < \chi_2$ will be a positive super-solution to \eqref{eq3:29jun12} given that 
$$
f(x,\chi) > q_1(\chi) = a_R\chi - \frac18\sigma^2\chi^{-7} -2\pi\hat{\rho}\chi^5.
$$ 
Similarly, we may choose
a positive sub-solution $\phi_-< \phi_+$ to \eqref{eq3:29jun12} by choosing any sufficiently small $\phi_-$ satisfying $0< \phi_- < \chi_3$, where $\chi_3$ is the lone positive root of
$$q_2(\chi) = a_R\chi- \frac18\sigma^2\chi^{-7}.$$ 
We can then apply the method of sub- and super-solutions 
outlined in Section~\ref{subsup} to solve \eqref{eq3:29jun12}.
\end{proof}

The next result extends Proposition~\ref{prop1:29jun12} to the case when $\la \ne 0$ and indicates that
$\rho_c$ is also a critical value for the decoupled problem \eqref{eq10:29jun12}.

\begin{corollary}\label{cor1:30jun12}
Let $\rho(x) \in C(\cM)$ and suppose that $\tau$ is a constant and that $$\rho_c = \frac{R^{\frac32}}{24\sqrt{3}|\sigma|\pi}.$$ 
There exists an $\ee > 0$ such that there is no positive solution to \eqref{eq10:29jun12} if $\rho > \rho_c$ and $-\ee < \la < 0$, and there exists a positive solution to \eqref{eq10:29jun12} if $0 < \rho \le \rho_c$ and $0 \le \la <  \ee$.
Finally, if $\rho = \rho_c$ and $\la$ is sufficiently small, then \eqref{eq10:29jun12} has a solution if and only if $ \la \ge 0$.
\end{corollary}
\begin{proof}
Again, we observe that if $\phi>0$ solves \eqref{eq10:29jun12}, then
\begin{align}\label{eq4:16july12}
\Delta \phi =  a_R\phi+\la^2a_{\tau}\phi^5-\frac18\sigma^2\phi^{-7}-2\pi\rho e^{-\la}\phi^5 = f(x,\phi,\la).
\end{align}
Let $q(\chi)$ be as in Corollary~\ref{cor1:29jun12} and define
$$p_1(\chi,\la) = a_R\chi+\la^2a_{\tau}\chi^5-\frac18\sigma^2\chi^{-7}-2\pi\check{\rho}e^{-\la}\chi^5,$$
where $\check{\rho} = \inf_{x\in\cM}\rho(x)$.
It is clear that $f(x,\phi,\la) \le p_1(\phi,\la)$ for any $\phi >0$, and for $\la < 0$ and $\rho > \rho_c$ we have that
\begin{align}\label{eq1:30jun12}
p_1(\chi,\la) &= a_R\chi +\la^2a_{\tau}\chi^5-\frac18\sigma^2\chi^{-7}-2\pi\check{\rho}e^{-\la}\chi^5\\
&\le  a_R\chi +(\la^2a_{\tau}-2\pi\rho_c+ 2\pi\rho_c \la +o(\la^2))\chi^5-\frac18\sigma^2\chi^{-7}\nonumber\\
&= q(\chi) + (\la^2a_{\tau}+ 2\pi\rho_c \la + o(\la^2))\chi^5 = q(\chi) +g(\la)\chi^5.\nonumber
\end{align}
Here we observe that $g(\la) \to 0$ as $\la \to 0$, and for $|\la|$ sufficiently small, $g(\la) < 0$ if $\la<0$.
By Proposition~\ref{prop1:16july12}, we know that if $\chi>0$ then $q(\chi)\le 0$.
So Eq.\ \eqref{eq1:30jun12} implies that if $\rho > \rho_c$ and $\la < 0$ is sufficiently small, then $f(x,\chi,\la) \le p_1(\chi,\la)< 0$, and the maximum principle then implies that \eqref{eq10:29jun12} will have no solution.

If $\rho \le \rho_c$, then define 
$$
p_2(\chi,\la) = a_R\chi +\la^2a_{\tau}\chi^5-\frac18\sigma^2\chi^{-7}-2\pi\hat{\rho}e^{-\la}\chi^5,
$$
where $\hat{\rho} = \sup_{x\in\cM}\rho(x)$.
It is clear that $f(x,\chi,\la) \ge p_2(\chi,\la)$ for all $\chi > 0$, and for $\la \le 0$ we have
\begin{align}\label{eq6:16july12}
p_2(\chi,\la) &= a_R\chi +\la^2a_{\tau}\chi^5-\frac18\sigma^2\chi^{-7}-2\pi\hat{\rho}e^{-\la}\chi^5\\
&\ge a_R\chi +(\la^2a_{\tau}-2\pi\rho_c+ 2\pi\rho_c \la +o(\la^2))\chi^5-\frac18\sigma^2\chi^{-7}\nonumber\\
&= q(\chi) + (\la^2a_{\tau}+ 2\pi\rho_c \la + o(\la^2))\chi^5 = q(\chi) +g(\la)\chi^5.\nonumber
\end{align}
Again, $g(\la) \to 0$ as $\la \to 0$ and $g(\la ) > 0$ for $\la> 0$ sufficiently small.
Therefore if
$\chi>0$, Eq.\ \eqref{eq6:16july12} implies that $f(x,\chi,\la) > p_2(\chi,\la) \ge q(\chi)$ if $\la \ge 0$.
The properties of $q(\chi)$ specified in Proposition~\ref{prop1:16july12} imply that
that for any $\la>0$, either $p_2(\chi,\la)$ has a single positive root $\chi_0$ and $p_2(\chi,\la) >0$ for all
$\chi > \chi_0$, or $~p_2(\chi,\la)$ has two distinct positive roots.
This implies that if $\la>0$ we can find a positive
super-solution $\phi_+$ to \eqref{eq10:29jun12}.
If $\la = 0$ we take $\phi_+ = \phi_c$ to be a super-solution where
$\phi_c$ is defined in Corollary~\ref{cor1:29jun12}.
Similarly, we can also find a positive sub-solution $\phi_-$
satisfying $\phi_-<\phi_+$ by choosing any sufficiently small $0< \phi_-<\chi_0$, where $\chi_0$ is the
unique positive root of 
$$r(\chi,\la) = a_R\chi +\la^2a_{\tau}\chi^5-\frac18\sigma^2\chi^{-7}.$$
The method of sub-and super-solutions outlined in Section~\ref{subsup} then implies that if
$\rho \le \rho_c$ and $\la \ge 0$, then \eqref{eq10:29jun12} has a solution.

Finally, we observe that if $\rho = \rho_c$, then we have that
$$
f(x,\chi,\la) = q(\chi)+ g(\la)\chi^5,
$$
where $f$ and $g$ are the same as above.
Therefore, when $\la$ is small and $\rho =\rho_c$,  
we can apply the above analysis to conclude that
\eqref{eq10:29jun12} will have a solution if and only if  $\la \ge 0$.
\end{proof}

\begin{remark}
We note that the negative sign in front of the term $2\pi\rho\chi^5$ in the polynomial
$$
q(\chi) = a_R\chi-\frac18\sigma^2\chi^{-7}-2\pi\rho\chi^5,
$$ 
played an essential role in allowing us to determine our critical density $\rho_c$ and critical solution $\phi_c$.
If this term were positive, then $q(\chi)$ would be monotonic increasing for $\chi>0$, and we would not be able to find a positive $\phi_c$ and $\rho_c$ so that $q(\phi_c) =0$ and $q'(\phi_c) = 0$.
As we saw in Corollary~\ref{cor1:30jun12} and Proposition~\ref{prop1:29jun12}, these properties of $q(\chi)$ played an important role in the existence of solutions to Eq.\ \eqref{eq10:29jun12} and Eq.\ \eqref{eq3:29jun12}.
Later in this article, we will also see that these properties of $q(\chi)$ play an important role in our non-uniqueness analysis by allowing for the kernel of the linearization of $F((\phi,\bw),\la)$ and
$G(\phi,\la)$ to be one-dimensional.
These facts further emphasize the role that terms with the ``wrong sign" (cf.~\cite{PY05}) have in the non-uniqueness phenomena associated with the CTS, CTT and XCTS formulations of the Einstein constraint equations.
\end{remark}

\subsection{Existence of a One Dimensional kernel of $D_XF((\phi_c,{\bf 0}),0)$ when $\rho= \rho_c$}\label{sec1:29jun12}

In the previous section we proved the existence of a critical density $\rho_c$ that affected whether Eq.\ \eqref{eq3:29jun12} and Eq.\ \eqref{eq10:29jun12} had positive solutions.
We now show that when $\rho = \rho_c$, the linearization of both \eqref{eq10:29jun12} and \eqref{eq6:27jun12} develops a one-dimensional kernel.

We first calculate the Fr\'{e}chet derivatives
$D_XF((\phi,\bw),\la)$ and $D_{\phi}G(\phi,\la)$.
To compute these derivatives, we need only compute the 
G\^{a}teaux derivatives given that the G-derivatives are continuous in a neighborhood of $((\phi_c,{\bf 0},0)$.
See \cite{EZ86} and Remark~\ref{rem1:28aug12}.
Therefore,
$$
D_XF((\phi_c, {\bf 0}),0) = \left.\frac{d}{dt}F((\phi_c + t\phi, t\bw),0)\right|_{t=0},
$$
where $(\phi,\bw) \in C^{2,\alpha}(\cM)\oplus C^{2,\alpha}(\cT\cM)$ satisfies $\|(\phi,\bw)\|_{C^{2,\alpha}(\cM)\oplus C^{2,\alpha}(\cT\cM)} = 1.$  

So for a given $((\phi,{\bf w}),\lambda)$, the Fr\'{e}chet derivative
$$
D_XF((\phi,\bw),\la): C^{2,\alpha}(\cM) \oplus C^{2,\alpha}(\cT\cM) \to C^{0,\alpha}(\cM) \oplus C^{0,\alpha}(\cT\cM),
$$ 
is a block matrix of operators where the first column
consists of derivatives of \\$F((\phi,\bw),\la)$ with respect to $\phi$ and the second column consists of derivatives
with respect to $\bw$.
This implies that 
\begin{align}\label{eq2:28jun12}
D_XF((\phi,\bw),\la) =  \left[ \begin{array}{cc} -\Delta +a_R +5\la^2 a_{\tau} \phi^4+7a_{\bw}\phi^{-8}-10\pi \rho_ce^{-\la}\phi^4 & \quad \overline{\mathbb{L}}\\
6\la b_{\tau}^a \phi^5 &\quad \mathbb{L}\end{array} \right],
\end{align}
where 
\begin{align}\label{eq1:28jun27}
\overline{\mathbb{L}}h = \overline{\mathbb{L}}(\phi,\bw) h =  -\frac{1}{4}\phi^{-7}\left( (\cL w)_{ab}(\cL h)^{ab}+\sigma_{ab}(\cL h)^{ab}\right),
\end{align}
and $\cL$ is the conformal Killing operator.
Similarly, in the CMC case the map 
$$
D_{\phi}G(\phi,\la):C^{2,\alpha}(\cM) \to C^{0,\alpha}(\cM),
$$
has the form
\begin{align}\label{eq11:29jun12}
D_{\phi}G(\phi,\la) =  -\Delta +a_R +5\la^2 a_{\tau} \phi^4+\frac78\sigma^2\phi^{-8}-10\pi \rho_ce^{-\la}\phi^4.
\end{align}
We now make some key observations about \eqref{eq2:28jun12}.
\begin{proposition}\label{prop3:27jun12}
Let $\phi_c$ be as in Corollary~\ref{cor1:29jun12}.
Then $F((\phi_c,{\bf 0}),0)=0$ and \\
$D_XF((\phi_c,{\bf 0}),0)$ has the form
\begin{align}\label{eq8:29jun12}
D_XF(\phi_c,{\bf 0},0) =  \left[ \begin{array}{cc} -\Delta & \tilde{\mathbb{L}}\\
0 & \mathbb{L}\end{array} \right],
\end{align}
where $\tilde{\mathbb{L}}:C^{k,\alpha}(\mathcal{TM}) \to C^{k-1,\alpha}(\mathcal{M})$ is defined by 
$$
\ol{\mathbb{L}}(\phi_c,{\bf 0})h = \tilde{\mathbb{L}}h = -\frac14\phi^{-7}_c\sigma_{ab}(\mathcal{L}h)^{ab},
$$ 
and $\mathcal{L}$ is the conformal killing operator.
\end{proposition}
\begin{proof}
By Corollary~\ref{cor1:29jun12} it follows that $\phi_c$ is a  root of the polynomial
$$q(\chi) = a_R\chi -\frac{1}{8}\sigma^2\chi^{-7} - 2\pi\rho_c\chi^5,$$
\noindent
and also a root of 
\begin{align}\label{eq3:28jun12}
q'(\chi) = a_R +\frac{7}{8}\sigma^2\chi^{-8} - 10\pi\rho_c\chi^4.
\end{align}
This implies that $F((\phi_c,{\bf 0}),0)=0$ and that Eq.\ \eqref{eq2:28jun12} reduces to \eqref{eq8:29jun12} when \\
${((\phi,\bw),\la) = ((\phi_c,{\bf 0}),0)}$.
\end{proof}

\begin{remark}\label{rem1:29jun12}
Corollary~\ref{cor1:29jun12} implies that \eqref{eq11:29jun12} reduces to 
\begin{align}\label{eq7:29jun12}
D_{\phi}G(\phi_c,0) = -\Delta,
\end{align}
in the CMC case.
Therefore $\text{ dim ker}(D_{\phi}G(\phi_c,0)))=1$ and it is spanned by the constant function $\phi = 1$.
\end{remark}

\begin{corollary}\label{cor1:28jun12}
Letting $\cH_1 = L^2(\cM)$ and $\cH_2 = L^2(\cT\cM)$, the $\cH_1\oplus \cH_2$-adjoint
of $D_XF((\phi_c,{\bf 0}),0)$ has the form
\begin{align}
(D_XF(\phi_c,{\bf 0},0))^* =  \left[ \begin{array}{cc} -\Delta & \quad 0\\
 \hat{\mathbb{L}} & \quad \mathbb{L}\end{array} \right],
\end{align}
where $\hat{\mathbb{L}}: C^{k,\alpha}(\cM)  \to C^{k-1,\alpha}(\cT\cM)$ is defined by
\begin{align}
\hat{\mathbb{L}}u = D^b(\frac14 \phi_c^{-7}u\sigma_{ab}).
\end{align}
\end{corollary}
\begin{proof}
Let $(u_1,\bv_1)$ and $(u_2,\bv_2)$ both be elements of $C^{2}(\cM)\oplus C^{2}(\cT \cM)$.
Then given that both $-\Delta$ and $\mathbb{L} = -D_b(\cL)^{ab}$ are self-adjoint with respect to
the $L^2(\cM)$ and $L^2(\cT\cM)$ inner products, it follows that
\begin{align}
\left\langle D_XF((\phi_c,{\bf 0}),0)\left[\begin{array}{c} u_1\\ \bv_1\end{array}\right],\left[\begin{array}{c} u_2\\ \bv_2\end{array}\right] \right\rangle 
= \int_{\cM} ( -u_1 \Delta u_2 +\bv_1\cdot\mathbb{L}\bv_2+ \tilde{\mathbb{L}}\bv_1 u_2 )dV_g,
\end{align}
where $dV_g$ is the volume element associated with $g_{ab}$ and $\tilde{\mathbb{L}}\bv_1 = -\frac14\phi^{-7}_c\sigma_{ab}(\mathcal{L}\bv_1)^{ab}$.
Given that the negative divergence of a $(0,2)$ tensor and the conformal killing operator $\mathcal{L}$ are formal adjoints (see \cite{EZ86}), we have that
\begin{align}\label{eq1:19july12}
&\int_{\cM} \tilde{\mathbb{L}}\bv_1 u_2 dV_g = \int_{\cM}\left(-\frac14 u_2\phi_c^{-7}\sigma_{ab}(\mathcal{L}\bv_1)^{ab}\right)dV_g\\
=& \int_{\cM}\left(D^{b}(\frac14 u_2 \phi^{-7}_c \sigma_{ab})\cdot \bv_1\right)dV_g = \int_{\cM}\hat{\mathbb{L}}u_2\cdot \bv_1 dV_g. \nonumber
\end{align}
Therefore,
\begin{align}
&\left\langle D_XF((\phi_c,{\bf 0}),0)\left[\begin{array}{c} u_1\\ \bv_1\end{array}\right],\left[\begin{array}{c} u_2\\ \bv_2\end{array}\right] \right\rangle =\\
&  \int_{\cM} ( -u_1 \Delta u_2 +\bv_1\cdot\mathbb{L}\bv_2+ \hat{\mathbb{L}}u_2\cdot \bv_1)dV_g
=\left\langle \left[\begin{array}{c} u_1\\ \bv_1\end{array}\right], \left[ \begin{array}{cc} -\Delta & \quad 0\\
 \hat{\mathbb{L}} & \quad \mathbb{L}\end{array} \right]\left[\begin{array}{c} u_2\\ \bv_2\end{array}\right] \right\rangle.\nonumber
\end{align}
\end{proof}

\begin{corollary}\label{cor1:27jun12}
$D_XF((\phi_c,{\bf 0},0)$ has a kernel of dimension 1 that is spanned by $\tiny{\left[\begin{array}{c} 1 \\ {\bf 0} \end{array} \right]}$,
and $(D_XF(\phi_c,{\bf 0},0))^*$ also has a kernel of dimension one that is spanned by $\tiny{\left[\begin{array}{c} 1 \\ {\bf 0}\end{array} \right]}$.
\end{corollary}
\begin{proof}
We solve for $\tiny{\left[\begin{array}{c} u\\ \bv\end{array}\right]}$ $\in C^{2,\alpha}(\cM)\oplus C^{2,\alpha}(\cT \cM)$ such that
$$
D_XF((\phi_c,{\bf 0}),0)\left[\begin{array}{c} u\\ \bv\end{array}\right] = \left[\begin{array}{cc}-\Delta& \tilde{\mathbb{L}}\\0& \mathbb{L}\end{array}\right]\left[\begin{array}{c}u \\ \bv\end{array}\right] = \left[\begin{array}{c}0\\{\bf 0}\end{array}\right].
$$
Given that $g_{ab}$ admits no conformal killing fields, we must have that $\bv = 0$.
This implies that
$$
0 = -\Delta u - \frac14\phi_c^{-7}(\sigma_{ab}(\cL\bv)^{ab})  = -\Delta u \Longrightarrow u \quad \text{is a constant}.
$$ 
Therefore $\tiny{\left[\begin{array}{c} 1 \\ {\bf 0} \end{array} \right]}$ spans $\text{ker}(D_XF((\phi_c,{\bf 0}),0)$.

Similarly, we solve for $\tiny{\left[\begin{array}{c} u\\ \bv\end{array}\right]}$ such that
$$
(D_XF((\phi_c,{\bf 0}),0))^*\left[\begin{array}{c} u\\ \bv\end{array}\right] = \left[\begin{array}{cc}-\Delta& 0\\ \hat{\mathbb{L}}& \mathbb{L}\end{array}\right]\left[\begin{array}{c}u \\ \bv\end{array}\right] = \left[\begin{array}{c}0\\{\bf 0}\end{array}\right].
$$
This implies that $u$ is a constant and that
$$
0= \hat{\mathbb{L}}u + \mathbb{L}\bv = \nabla^b(\frac14\phi_c u \sigma_{ab}) +\mathbb{L}\bv = \frac14\phi_c u \nabla^b\sigma_{ab} +\mathbb{L}\bv.
$$
Given that $\sigma_{ab}$ is divergence free, we have that $\nabla^b\sigma_{ab} = 0$, which implies that $\bv = 0$.
Therefore
$\tiny{\left[\begin{array}{c} 1 \\ {\bf 0} \end{array} \right]}$ spans $\text{ker}(D_XF((\phi_c,{\bf 0}),0)^*)$.
\end{proof}

We can now prove Theorems~\ref{thm4:29jun12} and \ref{thm2:29jun12}.
The proofs are an immediate consequence of the preceding results, but
we summarize them here in the proof for convenience.

\subsection{Proofs of Theorems~ \ref{thm4:29jun12} and \ref{thm2:29jun12}: Critical Parameter and Kernel Dimension}

Proposition~\ref{prop1:29jun12} implies the existence of critical values 
$$\rho_c = \left( \frac{R}{24\pi\rho_c}\right)^{\frac14} \quad \text{and} \quad  \phi_c = \left(\frac{a_R}{3\pi \rho_c}\right)^{\frac14},$$
such that if
$$q(\chi) = a_R\chi -\frac18\sigma^2\chi^{-7}-2\pi\rho_c\chi^5,$$ 
then $q(\phi_c) = q'(\phi_c) = 0$.
By Remark~\ref{rem1:29jun12} we have that the linearization
\eqref{eq11:29jun12} in the CMC case reduces to $-\Delta$.
This proves Theorem~\ref{thm4:29jun12}.
Similarly,
in Proposition~\ref{prop3:27jun12} we explicitly determined $D_XF((\phi_c,{\bf 0}),0)$, and in 
Corollary~\ref{cor1:29jun12} we showed that it has a kernel spanned by the constant vector $\tiny{\left[\begin{array}{c}1 \\ 0\end{array}\right]}$.
This proves Theorem~\ref{thm2:29jun12}.

\subsection{Fredholm properties of the operators $D_XF((\phi_c,{\bf 0}),0)$ and $D_{\phi}G(\phi_c,0)$}\label{sec1:26july12}

Now that we have shown that the linearizations $D_XF((\phi_c,{\bf 0}),0)$ and $D_{\phi}G(\phi_c,0)$ have one-dimensional
kernels, we are almost ready to apply the Liapunov-Schmidt reduction.
Recall from section~\ref{sec1:14july12} that a key assumption in this reduction was that the operator be a nonlinear Fredholm operator.
Therefore, to apply this reduction in the CMC and non-CMC cases we must show that the operators $D_{\phi}G(\phi_c,0)$ and $D_XF((\phi_c,{\bf 0}),0)$
are Fredholm operators between the spaces on which they are defined.
In particular, we need to show that $D_{\phi}G(\phi_c,0)$ is a Fredholm operator
between the spaces $C^{2,\alpha}(\cM)$ and $C^{0,\alpha}(\cM)$ and that the operator $D_XF((\phi_c,{\bf 0}),0)$ is a Fredholm operator between
$C^{2,\alpha}(\cM)\oplus C^{2,\alpha}(\cT\cM)$ and $C^{0,\alpha}(\cM)\oplus C^{0,\alpha}(\cT\cM)$.

In the CMC case, we have that $D_{\phi}G(\phi_c,0) = -\Delta$.
It is well known that this operator is a Fredholm operator between the
Hilbert spaces $H^2(\cM)$ and $L^2(\cM)$ \cite{HNT07b}.
Furthermore, $-\Delta$ is a Fredholm operator between the subspaces $C^{2,\alpha}(\cM)$
and $C^{0,\alpha}(\cM)$ because of the regularity properties of the the Laplacian and the fact that these spaces continuously embed into
the Hilbert spaces $H^2(\cM)$ and $L^2(\cM)$.
See Appendix~\ref{sec1:24july12} for a more detailed discussion of these facts.

Letting $L= -\Delta$, we regard $L= L^*$ as operators from $H^2(\cM) \to L^2(\cM)$.
The Fredholm properties of these operators
allow us to make the following decompositions that are orthogonal with respect to the $L^2$-inner product:
\begin{align}\label{eq1:31july12}
L^2(\cM) = R(L^*) \oplus \text{ker}(L)\\
L^2(\cM) = R(L) \oplus \text{ker}(L^*)\nonumber.
\end{align}
In this case, these decompositions are the same given that $L$ is self-adjoint.
Therefore if we regard $C^{2,\alpha}(\cM)$ and $C^{0,\alpha}(\cM)$ as subspaces of $L^2(\cM)$, 
then we may use \eqref{eq1:31july12} to obtain the following decompositions 
\begin{align}\label{eq2:31july12}
C^{2,\alpha}(\cM) = (R(L^*)\cap C^{2,\alpha}(\cM))\oplus  \text{ker}(L), \\
C^{0,\alpha}(\cM) = (R(L)\cap C^{0,\alpha}(\cM))\oplus \text{ker}(L^*),\nonumber
\end{align}
which are also orthogonal with respect to the $L^2$-inner product.
See Appendix~\ref{sec1:24july12}
for further details.

It is not as clear that the operator 
$D_XF((\phi_c,{\bf 0}),0)$ is a Fredholm operator between the spaces
$C^{2,\alpha}(\cM) \oplus C^{2,\alpha}(\cT\cM)$ and $ C^{0,\alpha}(\cM) \oplus C^{0,\alpha}(\cT\cM)$.
For the sake of completeness, we briefly discuss this point.
As in Appendix~\ref{sec1:24july12}, we first show that $D_XF(\phi_c,{\bf 0}),0)$ is a Fredholm operator from
the Hilbert space $L^2(\cM)\oplus L^2(\cT\cM)$ to itself, where we consider the domain of definition of 
$D_XF((\phi_c,{\bf 0}),0)$ to be 
${H^2(\cM)\oplus H^2(\cT\cM)}$.
Indeed, the operator $D_XF((\phi_c,{\bf 0}),0)$
induces the bilinear form 
$$B((u_1,\bv_1),(u_2,\bv_2)): (H^1(\cM)\oplus H^1(\cT\cM))\times (H^1(\cM)\oplus H^1(\cT\cM))\to \mathbb{R},$$
where $\langle \cdot, \cdot \rangle$ is the inner product associated with $L^2(\cM)\oplus L^2(\cT\cM)$ and
\begin{align}
B((u_1,\bv_1),(u_2,\bv_2))= \left\langle  \left[ \begin{array}{cc} -\Delta & \tilde{\mathbb{L}} \\ {\bf 0} & \mathbb{L}  \end{array}\right]\left[\begin{array}{c}u_1\\\bv_1\end{array}\right], \left[\begin{array}{c}u_2\\\bv_2\end{array}\right]\right\rangle .
\end{align}
Paralleling the discussion in \ref{sec1:24july12}, we first show there exists constants $C,c>0$ such that 
$$B((u,\bv),(u,\bv))+c\langle (u,\bv),(u,\bv)\rangle \ge C\|(u,\bv)\|^2_{H^1(\cM)\oplus H^1(\cT\cM)}.$$ 
Let $c>0$ be a constant to be determined.
Then
\begin{align}\label{eq1:25july12}
&B((u,\bv),(u,\bv))+c\langle (u,\bv),(u,\bv)\rangle \\
&= \int_{\cM}\left( D^auD_au - \frac14u \phi_c^{-7} \sigma_{ab}(\cL v)^{ab}+(\cL v)^{ab}(\cL v)_{ab} +cu^2+cv^av_a\right)dV_g\nonumber \\
&\ge  \int_{\cM}\left( D^auD_au - \frac{1}{16c\ee}u^2 - \ee\phi_c^{-14} (\sigma_{ab}(\cL v)^{ab})^2 + (\cL v)^{ab}(\cL v)_{ab} +cu^2+cv^av_a\right)dV_g, \nonumber
\end{align}
where the above inequality follows from an application of Young's inequality.
The Schwartz inequality and the definition of $\cL$ then imply that
$$
\sigma_{ab}(\cL v)^{ab} = \langle \sigma,\cL\bv\rangle_g \le C|\sigma||D\bv|.
$$
Therefore 
\begin{align}\label{eq2:25july12}
\int_{\cM} \ee\phi_c^{-14} (\sigma_{ab}(\cL v)^{ab})^2  \le c(\ee)\|\bv\|^2_{1,2},
\end{align}
where $c(\ee) \to 0$ as $\ee \to 0$.
Combining \eqref{eq1:25july12} and \eqref{eq2:25july12} we have that
\begin{align}
&B((u,\bv),(u,\bv))+c\langle (u,\bv),(u,\bv)\rangle \ge \\
& (1-c(\ee))\|\bv\|^2_{1,2} + \|Du\|^2_{0,2} + (c-\frac{1}{16\ee})\|u\|^2_{0,2} \ge C(\|\bv\|^2_{1,2}+\|u\|^2_{1,2}),\nonumber
\end{align}
where the final inequality holds by choosing $\ee$ sufficiently small and $c$ sufficiently large.

The above discussion tells us that the
bilinear form
$$B((u,\bv),(u,\bv))+c\langle (u,\bv),(u,\bv)\rangle $$
 is coercive on $H^1(\cM)\oplus H^1(\cT\cM)$.
The Lax-Milgram
theorem implies that the problem 
$$(D_XF((\phi_c,{\bf 0}),0)+cI)\left[\begin{array}{c} u \\ \bv \end{array}\right] = \left[\begin{array}{c} f \\ \bg \end{array}\right]$$
has a unique weak solution $(u,\bv) \in H^1(\cM)\oplus H^1(\cT\cM)$
for each $(f,\g) \in\\
 L^2(\cM)\oplus L^2(\cT\cM)$, and elliptic regularity gives us that $(u,\bv) \in H^2(\cM)\oplus H^2(\cT\cM)$.
Therefore we conclude that the operator $D_XF((\phi_c,{\bf 0}),0)+cI$ is a bijection between $H^2(\cM)\oplus H^2(\cT\cM)$ and $L^2(\cM)\oplus L^2(\cT\cM)$.
We are the able to conclude that 
$$(D_XF((\phi_c,{\bf 0}),0)+cI)^{-1} \quad \text{exists and is compact.}$$  
Paralleling the discussion in
Appendix~\ref{sec1:24july12}, we can then conclude that the operator $D_XF((\phi_c,{\bf 0}),0)$ is a Fredholm operator between
$H^2(\cM)\oplus H^2(\cT\cM)$ and \\
$L^2(\cM)\oplus L^2(\cT\cM)$.
Using the fact that $C^{0,\alpha}(\cM)\oplus C^{0,\alpha}(\cT\cM)$ embeds
continuously into  $L^2(\cM)\oplus L^2(\cT\cM)$ and invoking classical Schauder estimates, an argument similar to the argument in \ref{sec1:24july12}
implies that $D_XF((\phi_c,{\bf 0}),0)$ is Fredholm
operator between the spaces $C^{2,\alpha}(\cM)\oplus C^{2,\alpha}(\cT\cM)$ and $C^{0,\alpha}(\cM)\oplus C^{0,\alpha}(\cT\cM)$.
By applying the same argument to $D_XF((\phi_c,{\bf 0}),0)^*$, we can also conclude that this operator is a Fredholm operator between $C^{2,\alpha}(\cM)\oplus C^{2,\alpha}(\cT\cM)$ and $C^{0,\alpha}(\cM)\oplus C^{0,\alpha}(\cT\cM)$.

If $L= D_XF((\phi_c,{\bf 0}),0)$, then the fact that both $L, L^*$ are Fredholm operators from $H^2(\cM)\oplus H^2(\cT\cM) \to L^2(\cM)\oplus L^2(\cT\cM)$ allows us to decompose
${L^2(\cM) \oplus L^2(\cT\cM)}$ as in \eqref{eq1:31july12}.
Therefore, regarding $C^{2,\alpha}(\cM)\oplus C^{2,\alpha}(\cT\cM)$ and $C^{0,\alpha}(\cM)\oplus C^{0,\alpha}(\cT\cM)$ as subspaces of  $L^2(\cM) \oplus L^2(\cT\cM)$, we obtain the following decompositions that are orthogonal with respect to the
 $L^2(\cM) \oplus L^2(\cT\cM)$- inner product:
\begin{align}\label{eq3:31july12}
C^{2,\alpha}(\cM)\oplus C^{2,\alpha}(\cT\cM) = \text{ker}(L) \oplus (R(L^*)\cap (C^{2,\alpha}(\cM)\oplus C^{2,\alpha}(\cT\cM))),\\
C^{0,\alpha}(\cM)\oplus C^{0,\alpha}(\cT\cM) = \text{ker}(L^*) \oplus ( R(L)\cap (C^{0,\alpha}(\cM)\oplus C^{0,\alpha}(\cT\cM))).\nonumber
\end{align}
In the above decomposition, $L=  D_XF((\phi_c,{\bf 0}),0)$ and $\text{ker}(L), R(L), \text{ker}(L^*)$ and \\$R(L^*)$ are
all regarded as subspaces of $L^2(\cM) \oplus L^2(\cT\cM)$.

\section{Proofs of the Main Results}

\subsection{Proof of Theorem~\ref{thm3:29jun12}: Bifurcation and non-uniqueness in the CMC case}\label{sec4:14july12}

We are now ready to prove Theorem~\ref{thm3:29jun12}.
In the CMC case, our system \eqref{eq2:30apr12}
with $\rho= \rho_c$ reduces to
\begin{align}\label{eq12:2july12}
G(\phi,\la) = -\Delta\phi+ a_R\phi + \lambda^2 a_{\tau}\phi^5-\frac18\sigma^2\phi^{-7}-2\pi\rho_c e^{-\lambda}\phi^{5}.
\end{align}
To prove that solutions to \eqref{eq12:2july12} are non-unique, we will apply the Liapunov-Schmidt reduction outlined in Section~\ref{LSReduc}
and then invoke Theorem~\ref{thm1:29apr12} and Proposition~\ref{prop1:1july12}.

By Theorem~\ref{thm4:29jun12} and Remark~\ref{rem1:29jun12}, we know that $D_{\phi}G(\phi_c,0) = - \Delta.$
It follows that 
$\text{dim ker}(D_{\phi}G(\phi_c,0)) = \text{dim ker}(D_{\phi}G(\phi_c,0)^*) = 1,$
where both spaces are spanned by $\phi = 1$.

Using the notation from Section~\ref{LSReduc}, we can apply the Liapunov-Schmidt Reduction, where $\hat{v}_0 = 1$ is a basis of
$\text{ker}(D_{\phi}G(\phi_c,0)) = \text{ker}(D_{\phi}G(\phi_c,0)^*)$.
By the discussion in Section~\ref{sec1:26july12} and appendix~\ref{sec1:24july12}, we can decompose $X = C^{2,\alpha}(\cM) = X_1 \oplus X_2$ and $Y = C^{0,\alpha}(\cM) = Y_1\oplus Y_2$, where
\begin{align}
&X_1 =\text{ker}(D_{\phi}G(\phi_c,0)), \quad X_2= R(D_{\phi}G(\phi_c,0)^*)\cap C^{2,\alpha}(\cM),  \\
&Y_1 = R(D_{\phi}G(\phi_c,0))\cap C^{0,\alpha}(\cM), \quad \text{and} \quad Y_2 = \text{ker}(D_{\phi}G(\phi_c,0)^*).\nonumber
\end{align}
Letting $P:X\to X_1$ and $Q:Y\to Y_2 $ be projection operators as in Section~\ref{LSReduc}, and writing $\phi = P\phi + (I-P)\phi = v+w$, the Implicit Function Theorem applied to
\begin{align}\label{eq1:22aug12}
(I-Q)G(v+w,\la) = 0,
\end{align}
 implies that $w = \psi(v,\la)$ in a neighborhood of $(\phi_c,0)$ and $0 = \psi(\phi_c,0)$.
Plugging $\psi(v,\la)$ into
$$
QG(v+w,\la) = 0,
$$
we obtain
\begin{align}\label{eq3:19july12}
\Phi(v,\la) = QG(v+\psi(v,\la),\la) = 0.
\end{align}
All solutions to $G(\phi,\la) = 0$ in a neighborhood of $(\phi_c,0)$ must satisfy Eq.\ \eqref{eq3:19july12}.

We now observe that $D_{\la}G(\phi_c,0) = 2\pi\rho_c\phi_c^5 \ne 0$.
This implies that
\begin{align}
D_{\la}\Phi(\phi_c,0) = QD_{\la}G(\phi_c,0)  = 2\pi\rho_c\phi_c^5 \ne 0,
\end{align}
given that $Q$ is the projection onto $Y_2$ and $Y_2$ is spanned by the constant function $1$.
The Implicit Function Theorem applied to Eq.\ \eqref{eq3:19july12} implies that there exists a function $\gamma: U_1 \to V_1$ such that
$U_1\subset X_1$, $V_1 \subset \mathbb{R}$ and $\gamma(v) = \la$ in a neighborhood 
of $\phi_c$ with $\gamma(\phi_c) = 0$.

Therefore \eqref{eq3:19july12} becomes
\begin{align}\label{eq13:2july12}
g(v) =  QG(v+\psi(v,\gamma(v)),\gamma(v)),
\end{align}
and by writing $v = s + \phi_c$, which we can do for $s\in (-\delta,\delta)$ with $\delta>0$ sufficiently small, we obtain
\begin{align}\label{eq1:5oct12}
g(s) = QG(s+\phi_c + \psi(s+\phi_c,\gamma(s+\phi_c)),\gamma(s+\phi_c))=0.
\end{align} 
This implies that solutions to $G(\phi,\la) = 0$
are given by $g(s) = 0$ in a neighborhood of $(\phi_c,0)$, where
\begin{align}\label{eq14:2july12}
&\phi(s) = s+\phi_c +\psi(s+\phi_c,\gamma(s+\phi_c)), \\
&\la(s) = \gamma(s+ \phi_c) \nonumber
\end{align}
determine a differentiable solution curve through $(\phi_c,0)$.

Equation \eqref{eq14:2july12} gives us a fairly explicit representation of the continuously
differentiable curve $\{\phi(s),\la(s)\}$ provided by Theorem~\ref{thm1:29apr12}.
However,
by applying Proposition~\ref{prop1:1july12} we can determine that $\ddot{\la}(0) \ne 0$
to obtain even more information about $\{\phi(s),\la(s)\}$.
We observe that 
\begin{align}\label{eq15:2july12}
D^2_{\phi\phi}G(\phi_c,0)[\hat{v}_0,\hat{v}_0] = -7\sigma^2\phi_c^{-9}-40\pi\rho_c\phi_c^3 \ne 0.
\end{align}
Therefore  
$$-7\sigma^2\phi_c^{-9}-40\pi\rho_c\phi_c^3 \in Y_2  \Longrightarrow D^2_{\phi\phi}G(\phi_c,0)[\hat{v}_0,\hat{v}_0] \notin R(D_{\phi}G(\phi_c,0)) = Y_1,$$ 
given that $Y_1\perp Y_2$.
Proposition~\ref{prop1:1july12} implies that $\ddot{\la}(0) \ne 0$ and that a 
saddle node bifurcation occurs at $(\phi_c,0)$.

We now combine \eqref{eq14:2july12} and the fact that $\ddot{\la}(0) \ne 0$ to obtain a more
explicit representation to the solution curve $\{\phi(s),\la(s)\}$ in a neighborhood of $(\phi_c,0)$.
Define the function 
\begin{align}\label{eq15:6july12}
f(s) = \psi(s+\phi_c,\gamma(s+\phi_c)).
\end{align}
Then by Propositions~\ref{prop1:6july12} and \ref{prop1:1july12} we have that
\begin{align}\label{eq18:6july12}
&f(0) = 0,  \quad \text{and} \quad \la(0) = \gamma(\phi_c) = 0,\\
&\dot{\la}(0) = \left.\frac{d}{ds} \la(s)\right|_{s=0} = D_v\gamma(\phi_c) = 0, \nonumber \\
&\dot{f}(0) = \left.\frac{d}{ds} f(s)\right|_{s=0} = D_v\psi(\phi_c,0)+D_{\la}\psi(\phi_c,0)D_v\gamma(\phi_c) = 0.\nonumber
\end{align}
Therefore the function $f(s) = O(s^2)$.
By computing a Taylor expansion of $\la(s)$ about $s=0$ and using Eq.\ \eqref{eq18:6july12} and Eq.\ \eqref{eq14:2july12}, we find that for $s \in (-\delta,\delta)$,
\begin{align}\label{eq16:6july12}
&\phi(s) = \phi_c + s + O(s^2),\\
&\la(s) = \frac12\ddot{\la}(0)s^2 + O(s^3) \nonumber,
\end{align}
where $\ddot{\la}(0) \ne 0$.

Based on the form of $\phi(s)$ and $\la(s)$ in Eq.\ \eqref{eq16:6july12}, there exists a $\delta' \in (0,\delta)$ such that $\phi(s) <0$, $\la(s) >0$
for all $s \in [-\delta',0)$,  and $\phi(s) >0$, $\la(s) >0$ for all $s \in (0,\delta']$.
Letting $M = \min\{M_1,M_2\}$, where 
$$
M_1 = \sup_{s \in [-\delta',0]} \la(s) \quad \text{and} \quad M_2  =  \sup_{s \in [0,\delta']} \la(s),
$$
the Intermediate Value Theorem then implies that for all $\la_0 \in (0,M)$, there exists $s_1, s_2 \in [-\delta', \delta']$, $s_1 \ne s_2$, such that $\la(s_1)= \la(s_2) = \la_0$.
Based on how we chose $\delta'$, we also have that $\phi(s_1) \ne \phi(s_2)$.
This completes the proof of Theorem~\ref{thm3:29jun12}.

\subsection{Proof of Theorem~\ref{thm1:29jun12}: Bifurcation and non-uniqueness in the non-CMC case}\label{sec5:14july12}

In this section we will show that solutions to $F((\phi,\bw),0) = 0$ for the full system 
\begin{align}\label{eq4:31july12}
F((\phi,{\bf w}),\lambda) = \left[ \begin{array}{c} -\Delta \phi + a_R\phi + \lambda^2 a_{\tau}\phi^5-a_{{\bf w}}\phi^{-7}-2\pi\rho e^{-\lambda}\phi^{5}\\
\mathbb{L}{\bf w} + \lambda b_{\tau}^a\phi^6 
 \end{array} \right]
\end{align}
are non-unique, where $\tau \in C^{1,\alpha}(\cM)$ is a non-constant function.
Our approach is similar to that of the CMC case: we apply a Liapunov-Schmidt reduction to Eq.\ \eqref{eq4:31july12} to determine
an explicit solution curve through the point $((\phi_c,{\bf 0}),0)$.
The form of this curve will imply that solutions to the system \eqref{eq4:31july12} are non-unique.

By Proposition~\ref{prop3:27jun12} we know that $\text{ker}D_XF((\phi_c,{\bf 0}),0)$
takes the form
$$
D_XF(\phi_c,{\bf 0},0) =  \left[ \begin{array}{cc} -\Delta & \tilde{\mathbb{L}}\\
0 & \mathbb{L}\end{array} \right],
$$
where $\tilde{\mathbb{L}}h = -\frac14\phi^{-7}_c\sigma_{ab}(\mathcal{L}h)^{ab}$.
Corollary~\ref{cor1:27jun12} gives us that
$\text{ker}(D_XF((\phi_c,{\bf 0}),0))$ and $\text{ker}(D_XF((\phi_c,{\bf 0}),0)^*)$ are spanned by $\hat{v}_0 =\tiny{\left[\begin{array}{c}1\\ 0 \end{array}\right]}$.

Using the notation from Section~\ref{LSReduc}, we apply the Liapunov-Schmidt Reduction.
By the decomposition \eqref{eq3:31july12}, we have that 
$$
X = C^{2,\alpha}(\cM)\oplus C^{2,\alpha}(\cT \cM) = X_1 \oplus X_2 ,
$$
and
$$
Y = C^{0,\alpha}(\cM)\oplus C^{0,\alpha}(\cT \cM) = Y_1\oplus Y_2,
$$ 
where
\begin{align}\label{eq5:31july12}
&X_1 = \text{ker}(D_XF((\phi_c,{\bf 0}),0)), \\
&X_2 =  R(D_XF((\phi_c,{\bf 0}),0)^*)\cap (C^{2,\alpha}(\cM)\oplus C^{2,\alpha}(\cT\cM)),\\
&Y_1 = R(D_XF((\phi_c,{\bf 0}),0))\cap (C^{0,\alpha}(\cM)\oplus C^{0,\alpha}(\cT\cM)), \hspace{2mm}\\ 
&Y_2 = \text{ker}(D_XF((\phi_c,{\bf 0}),0)^*).
\end{align}
\medskip

Let $P:X\to X_1$ and $Q:Y\to Y_2$ be the projection operators defined using $\hat{v}_0$ as in Section~\ref{LSReduc}.
Then by writing 
$$\left[ \begin{array}{c}\phi\\ \bw\end{array}\right]= P\left[ \begin{array}{c}\phi\\ \bw\end{array}\right]+ (I-P)\left[ \begin{array}{c}\phi\\ \bw\end{array}\right] = v+y,$$ 
the Implicit Function Theorem applied to
\begin{align}\label{eq2:22aug12}
(I-Q)F(v+y,\la) = 0,
\end{align}
implies that solutions to $F((\phi,\bw),\la)=0$ satisfy
\begin{align}\label{eq2:19july12}
\Phi(v,\la) = QF(v+\psi(v,\la),\la) = 0
\end{align}
in a neighborhood of $((\phi_c,{\bf 0}),0)$, where $y = \psi(v,\la)$ in this neighborhood and \\
$(0,{\bf 0}) = \psi((\phi_c,{\bf 0}),0)$.

We now observe that 
$$
D_{\la}F((\phi_c,{\bf 0}),0) = \left[\begin{array}{c}2\pi\rho_c\phi_c^5 \\ b_{\tau}^a\phi_c^6\end{array}\right]\notin Y_1,
$$
due to the fact that
$$
 \left[\begin{array}{c}2\pi\rho_c\phi_c^5 \\ 0\end{array}\right] \in Y_2  \quad \text{and} \quad Y_1 \perp Y_2.
$$  
This implies that
\begin{align}
D_{\la}\Phi((\phi_c,{\bf 0}),0) = QD_{\la}F((\phi_c,{\bf 0}),0) =  \left[\begin{array}{c}2\pi\rho_c\phi_c^5 \\ 0\end{array}\right]\ne 0,
\end{align}
given that $Q$ is the projection onto $Y_2$.
The Implicit Function Theorem again implies that there exists a function $\gamma: U_1 \to V_1$, where
$(\phi_c,{\bf 0}) \in U_1\subset X_1$, $V_1 \subset \mathbb{R}$ and $\gamma(v) = \la$ in $U_1$
with $\gamma(\phi_c,{\bf 0}) = 0$.
Using this fact, Eq.\ \eqref{eq2:19july12} becomes
\begin{align}\label{eq2:5oct12}
g(v) =  QF(v+\psi(v,\gamma(v)),\gamma(v))=0,
\end{align}
and by writing 
$$v = (s+\phi_c)\hat{v}_0 =  s\left[\begin{array}{c}1\\0\end{array}\right] + \left[\begin{array}{c} \phi_c\\0\end{array}\right],$$ 
for $s\in (-\delta,\delta)$ with $\delta>0$ sufficiently small, we then obtain
\begin{align}\label{eq1:5july12}
g(s) = QF\left(s\hat{v}_0 +\phi_c\hat{v}_0+\psi(s\hat{v}_0 +\phi_c\hat{v}_0,\gamma(s\hat{v}_0+\phi_c\hat{v}_0)),\gamma(s\hat{v}_0+\phi_c\hat{v}_0) \right) = 0.
\end{align} 
This implies that solutions to $F((\phi,{\bf 0}),\la) = 0$ in a neighborhood of $((\phi_c,{\bf 0}),0)$
satisfy $g(s) = 0$, where
{\small
\begin{align}\label{eq2:5july12}
\left[\begin{array}{c}\phi(s)\\ \bw(s)\end{array}\right] &= s\left[\begin{array}{c}1\\0\end{array}\right] +\left[\begin{array}{c}\phi_c\\0\end{array}\right]  +\psi\left(s\left[\begin{array}{c}1\\0\end{array}\right] +\left[\begin{array}{c}\phi_c\\0\end{array}\right] ,\gamma\left(s\left[\begin{array}{c}1\\0\end{array}\right] +\left[\begin{array}{c}\phi_c\\0\end{array}\right]\right) \right),\\
\la(s) &= \gamma\left(s\left[\begin{array}{c}1\\0\end{array}\right] + \left[\begin{array}{c}\phi_c\\0\end{array}\right] \right) ,\nonumber
\end{align}}
determine a smooth solution curve through $((\phi_c,{\bf 0}),0)$.

As in the CMC case, we seek additional information so that we can further analyze
the solution curve \eqref{eq2:5july12}.
Now we apply Proposition~\ref{prop1:1july12} to determine information about $\ddot{\la}(0)$, and
then we will expand the function
\begin{align}\label{eq4:5july12}
f(s) = \psi((s+\phi_c)\hat{v}_0,\gamma((s+\phi_c)\hat{v}_0))
\end{align}
as a Taylor series to obtain a more explicit representation of $\{(\phi(s),\bw(s)),\la(s)\}$.
  
Taking the second derivative of $F((\phi,\bw),\la)$, we have that 
\begin{align}\label{eq3:5july12}
D^2_{XX}F((\phi_c,{\bf 0}),0)[\hat{v}_0,\hat{v}_0] = \left[\begin{array}{c} -7\sigma^2\phi_c^{-9}-40\pi\rho_c\phi_c^3 \\ {\bf 0}\end{array}\right] \in Y_2.
\end{align}
Given that the vector \eqref{eq3:5july12} lies in $Y_2$ and $Y_1\perp Y_2$,
$$
D^2_{XX}F((\phi_c,{\bf 0}),0)[\hat{v}_0,\hat{v}_0] \notin Y_1.
$$
We can therefore apply Proposition~\ref{prop1:1july12} to conclude that $\ddot{\la}(0) \ne 0$.

Our next goal is to expand the function $f(s)$ as a Taylor series about $0$.
In order to do this,
we use \eqref{eq2:5july12}, Proposition~\ref{prop1:6july12} and the fact that $\ddot{\la}(0) \ne 0$ to obtain
information about coefficients in this expansion.
In particular, the objective is to determine
information about the coefficient of the second order term in the expansion of $f(s)$.

By differentiating $$(I-Q)F(v+\psi(v,\la),\la) = 0,$$ with respect to $\la$ and evaluating the resulting expression
at $((\phi_c,{\bf 0}),0),$ we obtain
\begin{align}\label{eq9:5july12}
(I-Q)D_XF((\phi_c,{\bf 0}),0)D_{\la}\psi((\phi_c,{\bf 0}),0)+(I-Q)D_{\la}F((\phi_c,{\bf 0}),0) = 0.
\end{align}
Given that 
$$D_{\la}F((\phi_c,{\bf 0}),0) =\left[\begin{array}{c} 2\pi\rho_c\phi_c^5 \\ b_{\tau}^a\phi_c^6\end{array}\right],$$
and $Q$ is the projection operator onto $Y_2$, which is spanned by $\left[\begin{array}{c}1\\{\bf 0}\end{array}\right]$,
we have that
\begin{align}\label{eq10:5july12}
(I-Q)D_{\la}F((\phi_c,{\bf 0}),0) = \left[\begin{array}{c} 0 \\ b_{\tau}^a\phi_c^6\end{array}\right].
\end{align}
Equations \eqref{eq10:5july12} and \eqref{eq9:5july12} imply that
\begin{align}\label{eq8:5july12}
(I-Q)D_XF((\phi_c,{\bf 0}),0)D_{\la}\psi((\phi_c,{\bf 0}),0)= - \left[\begin{array}{c} 0 \\ b_{\tau}^a\phi_c^6\end{array}\right].
\end{align}
Given that $D_XF((\phi_c,{\bf 0}),0)$ has the form \eqref{eq8:29jun12} and the operator $\mathbb{L}$ is invertible,
Eq.\ \eqref{eq8:5july12} implies that 
\begin{align}\label{eq11:5july12}
D_{\la}\psi((\phi_c,{\bf 0}),0) = \left[\begin{array}{c}u(x)\\\bv(x) \end{array}\right], \quad \text{with}\quad \bv(x) \ne {\bf 0}.
\end{align}
As we shall see, this fact implies that $\bw(s)$ has quadratic terms in $s$.

We have one last piece of data left to determine the coefficient of the second order term in the
Taylor expansion of $f(s)$.
Differentiating $(I-Q)F(v+\psi(v,\la),\la) = 0$ twice with respect to $v$, evaluating at $((\phi_c,{\bf 0}),0)$ and applying the resulting bilinear form to $\hat{v}_0$, we obtain
\begin{align}\label{eq12:5july12}
&(I-Q)D^2_{XX}F((\phi_c,{\bf 0}),0)[\hat{v}_0,\hat{v}_0]+\\
&(I-Q)D_XF((\phi_c,{\bf 0}),0)D^2_{vv}\psi((\phi_c,{\bf 0}),0)[\hat{v}_0,\hat{v}_0] = 0.\nonumber
\end{align} 
By Eq.\ \eqref{eq3:5july12} we know that $D^2_{XX}F((\phi_c,{\bf 0}),0)[\hat{v}_0,\hat{v}_0] \in Y_2$.Because $(I-Q)$ projects
onto $Y_1 $ and $Y_1 \perp Y_2$, we have that 
\begin{align}\label{eq13:5july12}
(I-Q)D^2_{XX}F((\phi_c,{\bf 0}),0)[\hat{v}_0,\hat{v}_0] = 0.
\end{align}
Equations \eqref{eq13:5july12} and \eqref{eq12:5july12} and the invertibility of $(I-Q)D_XF((\phi_c,{\bf 0}),0)$
as an operator from $X_2$ to $Y_1$ imply that 
\begin{align}\label{eq14:5july12}
D^2_{vv}\psi((\phi_c,{\bf 0}),0)[\hat{v}_0,\hat{v}_0] = 0.
\end{align}
This was the final piece of information that we needed to to determine the second order expansion of $f(s)$.

We now expand the function $f(s)$ in Eq.\
\eqref{eq4:5july12} about $s= 0$.
We have that
\begin{align}\label{eq5:5july12}
&f(0) = \psi((\phi_c,{\bf 0}),0 ) =  \left[\begin{array}{c}0\\{\bf 0}\end{array}\right], \\
&\dot{f}(0) = D_v\psi((\phi_c,{\bf 0}),0)\hat{v}_0+ D_{\la}\psi((\phi_c,{\bf 0}),0))D_v\gamma(\phi_c,{\bf 0})\hat{v}_0  = \left[\begin{array}{c}0\\{\bf 0}\end{array}\right],\nonumber\\
&\ddot{f}(0) = D^2_{vv}\psi((\phi_c,{\bf 0}),0)[\hat{v}_0,\hat{v}_0] +  D^2_{v\la}\psi((\phi_c,{\bf 0}),0)[\hat{v}_0,D_v\gamma(\phi_c,{\bf 0})\hat{v}_0] \nonumber\\
&+  D^2_{\la v}\psi((\phi_c,{\bf 0}),0)[D_v\gamma(\phi_c,{\bf 0})\hat{v}_0,\hat{v}_0]+   D_{\la }\psi((\phi_c,{\bf 0}),0)D^2_{vv}\gamma(\phi_c,{\bf 0})[\hat{v}_0,\hat{v}_0] \nonumber\\
&+  D^2_{\la \la}\psi((\phi_c,{\bf 0}),0)[D_v\gamma(\phi_c,{\bf 0})\hat{v}_0,D_v\gamma(\phi_c,{\bf 0})\hat{v}_0] \nonumber\\
&= D_{\la }\psi((\phi_c,{\bf 0}),0)D^2_{vv}\gamma(\phi_c,{\bf 0})[\hat{v}_0,\hat{v}_0] = D_{\la}\psi((\phi_c,{\bf 0}),0)\ddot{\la}(0) \ne \left[\begin{array}{c}0\\{\bf 0}\end{array}\right],\nonumber
\end{align}
where $\ddot{f}(0)$ simplifies as a result of Proposition~\ref{prop1:1july12}, Eq.\ \eqref{eq11:5july12} and Eq.\ \eqref{eq14:5july12}, which imply
\begin{align}\label{eq4:5oct12}
&D_v\psi((\phi_c,{\bf 0}),0) = 0,     &D_v\gamma(\phi_c,{\bf 0}) = 0, \quad \quad \quad \\ 
&D^2_{vv}\psi((\phi_c,{\bf 0}),0)[\hat{v}_0,\hat{v}_0] = \left[\begin{array}{c} 0 \\ {\bf 0}\end{array}\right], & D_{\la}\psi((\phi_c,{\bf 0}),0) \ne \left[\begin{array}{c} 0 \\ {\bf 0}\end{array}\right]. \nonumber
\end{align}
Therefore it follows that 
\begin{align}\label{eq6:5july12}
f(s) =  \frac12(D_{\la}\psi(\phi_c\hat{v}_0,\gamma(\phi_c\hat{v}_0))\ddot{\la}(0))s^2+O(s^3) = \left[\begin{array}{c}\frac12 u(x)\ddot{\la}(0)\\ \frac12 \bv(x)\ddot{\la}(0)\end{array}\right]s^2 + O(s^3),
\end{align}
where we identify $D_{\la}\psi((\phi_c,{\bf 0}),0) $ with the vector $\tiny{ \left[\begin{array}{c} u(x) \\\bv(x) \end{array}\right]}$ in $C^{2,\alpha}(\cM)\oplus C^{2,\alpha}(\cT \cM)$.
By Eq.\ \eqref{eq11:5july12} we have that $\bv(x) \ne 0$ and expanding out $\la(s)$ as a second order
Taylor series about $s =0$ we obtain
\begin{align}\label{eq16:5july12}
\la(s) = \frac12\ddot{\la}(0)s^2 + O(s^3).
\end{align}
Putting together \eqref{eq2:5july12}, \eqref{eq6:5july12} and \eqref{eq16:5july12} we find that
solutions to $F((\phi,\bw),\la) = 0$ in a neighborhood of $((\phi_c,{\bf 0}),0)$ take the form
\begin{align}
&\phi(s) = \phi_c + s + \frac12 \ddot{\la}(0) u(x) s^2 + O(s^3), \label{eq17:5july12}\\
&\bw(s) = \frac12 \ddot{\la}(0) \bv(x) s^2 + O(s^3), \label{eq2:28aug12}\\
&\la(s) = \frac12\ddot{\la}(0)s^2+O(s^3),\label{eq3:28aug12}
\end{align}
where $s \in (-\delta,\delta)$ for sufficiently small $\delta>0$.

By analyzing the solution curve \eqref{eq17:5july12}-\eqref{eq3:28aug12} as we did for the curve \eqref{eq16:6july12} in the proof of Theorem~\ref{thm3:29jun12}, we can conclude that solutions to the system \eqref{eq4:31july12} are non-unique.
This completes the proof of Theorem~\ref{thm1:29jun12}.

\section{Summary}\label{sec6:14july12}

We began in Section~\ref{sec1:14july12} by introducing our notation for function spaces and presenting the basic
concepts from functional analysis and bifurcation theory that we used throughout this paper.
In particular, we gave an outline of the Liapunov-Schmidt reduction that was the basis of our non-uniqueness arguments.
Then in Section~\ref{sec2:14july12} we presented our main results, which consisted of the existence of a critical solution
where the linearizations of our system 
\begin{align}\label{eq1:26july12}
 F((\phi,\bw),\la) = \left[ \begin{array}{c}-\Delta \phi + a_R\phi + \lambda^2a_{\tau}\phi^5-a_{{\bf w}}\phi^{-7}-2\pi\rho e^{-\lambda}\phi^{5}\\
\mathbb{L}{\bf w} + \lambda b_{\tau}^a\phi^6 \end{array}\right],
\end{align}
developed a one-dimensional kernel and non-uniqueness results for solutions to \\
$F((\phi,\bw),0) = 0$ in both
the CMC and non-CMC cases.
We then set about proving these results in the following sections.
In Section~\ref{sec3:14july12}
we showed that in the CMC case there exists a critical density $\rho_c$ for the operator 
\begin{align}\label{eq2:26july12}
G(\phi,\la) = -\Delta \phi + a_R\phi +\la^2a_{\tau}- a_{{\bf w}}\phi^{-7}-2\pi\rho e^{-\lambda}\phi^{5}.
\end{align}
This density satisfied the property that if $|\la|$ was sufficiently small, then $\rho > \rho_c$ and $\la < 0$ implied
that there was no solution to $G(\phi,\la) = 0$, and if $\rho\le \rho_c$ and $\la \ge 0$ then there was a solution.
This result provided the foundation
in Section~\ref{sec1:29jun12} for showing that the linearization of \eqref{eq1:26july12} developed a one-dimensional kernel.
Then in Section~\ref{sec1:26july12}
we briefly discussed the Fredholm properties of the linearized operators $D_XF((\phi_c,{\bf 0}),0)$ and $D_{\phi}G(\phi_c,0)$ on the
Banach spaces on which they are defined.

In Section~\ref{sec4:14july12} we proved the first of our non-uniqueness results.
We showed that in the event that the mean
curvature was constant, the decoupled system \eqref{eq2:26july12} exhibited non-uniqueness.
This was indicated by the fact that the solution
curve through the point $(\phi_c,0)$ had the form 
\begin{align}\label{eq3:26july12}
\phi(s) = \phi_c + s + O(s^2),\\
\la(s) = \frac12 \ddot{\la}(0)s^2 + O(s^3)\nonumber,
\end{align}
which implied that a saddle-node bifurcation occurred at the point $(\phi_c,0)$.
We were able to determine the explicit form of the solution curve
\eqref{eq3:26july12} by applying a Liapunov-Schmidt reduction to \eqref{eq2:26july12} at the point $(\phi_c,0)$, which was possible given
that the operator \\
$D_{\phi}G(\phi_c,0)$ had a one-dimensional kernel.
Similarly, in Section~\ref{sec5:14july12} we showed that
when the mean curvature $\tau$ was an arbitrary, continuously differentiable function, solutions to $F((\phi,\bw),\la) =0$ were non-unique.
Again, this followed because we explicitly computed the solution curve through the point $((\phi_c,{\bf 0}),0)$.
In Section~\ref{sec5:14july12} we found that the solution curve through $((\phi_c,{\bf 0}),0)$ had the form
\begin{align}
&\phi(s) = \phi_c + s + \frac12 \ddot{\la}(0) u(x) s^2 + O(s^3),\\
&\bw(s) = \frac12 \ddot{\la}(0) \bv(x) s^2 + O(s^3),\\
&\la(s) = \frac12\ddot{\la}(0)s^2+O(s^3),
\end{align}
which we demonstrated by applying a Liapunov-Schmidt reduction to the system \eqref{eq1:26july12} at the point 
$((\phi_c,{\bf 0}),0)$.
Again, this was possible because of our work in Section~\ref{sec3:14july12} where we showed that the
linearization $D_XF((\phi_c,{\bf 0}),0)$ had a one-dimensional kernel.

The importance of these non-uniqueness results is that they demonstrate first and foremost that the conformal formulation
with unscaled source terms is undesirable given that solutions for this formulation will not allow us to uniquely parametrize physical solutions
to the Einstein constraint equations.
Additionally, this paper helps build on the work of Walsh in \cite{DW07} by expanding the understanding of
how bifurcation techniques can be applied to the various conformal formulations of the constraint equations.
This work is also interesting in that
the analysis conducted here helps clarify the ideas of Baumgarte, O'Murchadha, and Pfeiffer in \cite{BOP07} by
showing how terms with ``the wrong sign" that contribute to the non-monotonicity (non-convexity of the corresponding energy) of the nonlinearity in the Hamiltonian constraint directly contribute to the non-uniqueness of solutions.
Finally, it is hope of the authors that this work will also help to lay the foundation for future analysis of the uniqueness properties of
the Conformal Thin Sandwich method and the far-from-CMC solution framework established in \cite{HNT07a,HNT07b}.

\section{Appendix}\label{sec7:14july12}

\subsection{Banach Calculus and the Implicit Function Theorem}\label{sec1:10july12}

Here we give a brief review of some basic tools from functional analysis.
The following results are
presented without proof and are taken from \cite{EZ86}.
We begin with some notation.

Suppose that $X$ and $Y$ are Banach spaces and $U \subset X$ is a neighborhood of $0$.
For a given map $f:U\subset X \to Y $, we say that 
$$
f(x) = o(\|x\|), \hspace{2mm} x\to 0 \quad \text{iff} \hspace{3mm} r(x)/\|x\| \to 0 \hspace{2mm} \text{as} \hspace{2mm} x\to 0.
$$
We write $L(X,Y)$ for the class of continuous linear maps between the Banach spaces $X$ and $Y$.

\begin{definition}
Let $U \subset X$ be a neighborhood of $x$ and suppose that $X$ and $Y$ are Banach spaces.
\begin{itemize}

\item[(1)] We say that a map $f:U \to Y$ is {\bf F-differentiable} or {\bf Fr\'{e}chet differentiable} at $x$ iff
there exists a map $T\in L(X,Y)$ such that
$$
f(x+h) - f(x) = Th+o(\|h\|), \quad \text{as} \hspace{2mm} h \to 0,
$$
for all $h$ in some neighborhood of zero.
If it exists, $T$ is called the {\bf F-derivative} or {\bf Fr\'{e}chet derivative}
of $f$ and we define $f'(x) = T$.
If $f$ is Fr\'{e}chet differentiable for all $x\in U$ we say that $f$ is Fr\'{e}chet differentiable
in $U$.
Finally, we define the {\bf F-differential} at $x$ to be $df(x;h) = f'(x)h$.

\item[(2)] The map $f$ is {\bf G-differentiable} or {\bf G\^{a}teaux differentiable} at $x$ iff there exists a map
$T\in L(X,Y)$ such that
$$
f(x+tk)-f(x) = tTk +o(t), \quad \text{as} \hspace{2mm} t \to 0,
$$
for all $k$ with $\|k\|=1$ and all real numbers $t$ in some neighborhood of zero.
If it exists, $T$ is called 
the {\bf G-derivative} or {\bf G\^{a}teaux derivative} of $f$ and we define $f'(x)=T$.
If $f$ is G-differential for all $x\in U$ we say that $f$ is G\^{a}teaux differentiable in $U$.
The {\bf G-differential} at $x$ is defined to be
$d_Gf(x;h) = f'(x)h.$
\end{itemize}
\end{definition}

\begin{remark}\label{rem1:28aug12}
Clearly if an operator is F-differentiable, then it must also be \\
G-differentiable.
Moreover, if the G-derivative $f'$ exists in some neighborhood of $x$ and $f'$ is continuous at $x$, then $f'(x)$ is also the F-derivative.
This fact is quite useful for computing F-derivatives given that G-derivatives are easier to compute.
See \cite{EZ86} for a complete discussion.
\end{remark}

\noindent
We view F-derivatives and G-derivatives as linear maps $f'(x) : U \to L(X,Y)$.
More generally, we may consider higher order derivatives maps of $f$.
For example, the map ${f''(x):U \to L(X,L(X,Y))}$ is a bilinear form.
We now state some basic properties of F-derivatives.
All of the following properties also hold for G-derivatives.

The Fr\'{e}chet derivative satisfies many of the usual properties that we are accustomed to by doing calculus in $\mathbb{R}^n$.
For example, we have the chain rule.

\begin{proposition}[Chain Rule]\label{chainRule}
Suppose that $X, Y$ and $Z$ are Banach spaces and assume that $f:U \subset X\to Y$ and $g:V\subset Y \to Z$
are differentiable on $U$ and $V$ resp.\ and that $f(U) \subset V$.
Then the function $H(x) = g \circ f $, i.e.\ $H(x) = g(f(x))$, is differentiable where
$$
H'(x) = g'(f(x))f'(x)
$$
where we write $g'(f(x))f'(x)$ for $g'(f(x))\circ f'(x)$.
\end{proposition}

Given an operator $f:X\times Y \to Z$, we can also consider the partial derivative
of $f$ with respect to either $x$ or $y$.
If we fix the variable $y$ and define
$g(x) = f(x,y): X \to Z$ and $g(x)$ is Fr\'{e}chet differentiable at $x$, then the
{\bf partial derivative} of $f$ with respect to $x$ at $(x,y)$ is $f_x(x,y) = g'(x)$.
We can a make a similar definition for $f_y(x,y)$.
Finally, we observe that
we can express the F-differential of $f'(x,y)$ in terms of the partials by using
the following formula:
\begin{align}\label{eq1:7july12}
f'(x,y)(h,k) = f_x(x,y)h + f_y(x,y)k.
\end{align}
We have the following relationship between the partial derivatives and the
Fr\'{e}chet \\derivative.
\begin{proposition}
Suppose that $f:X \times Y \to Z$ is F-differentiable at $(x,y)$.
Then the partial F-derivatives
$f_x$ and $f_y$ exist at $(x,y)$ and they satisfy \eqref{eq1:7july12}.
Moreover, if 
$f_x$ and $f_y$ both exist and are continuous in a neighborhood of $(x,y)$ then $f'(x,y)$ exists as an
F-derivative and \eqref{eq1:7july12} holds.
\end{proposition}

\subsubsection{Implicit Function Theorem}

Suppose that $F:U\times V \to Z$ is a mapping with $U\subset X, V\subset Y$ and $X,Y,Z$ are real
Banach spaces.
The {\bf Implicit Function Theorem} is an extremely important tool in analyzing
the nonlinear problem
\begin{align}\label{eq1:2july12}
F(x,y) = 0.
\end{align}
We present the statement of the Theorem here, the form of which is taken from \cite{HK04}.
For a proof see \cite{EZ86,CH82}.
\begin{theorem}\label{thm1:2july12}
Let \eqref{eq1:2july12} have a solution $(x_0,y_0) \in U\times V$ such that the Fr\'{e}chet derivative of $F$ with respect
to $x$ at $(x_0,y_0)$ is bijective:
\begin{align}
&F(x_0,y_0) = 0,\\
&D_xF(x_0,y_0):\to Z \quad \text{is bounded (continuous)}\nonumber\\
&\text{with bounded inverse.}\nonumber
\end{align}
Assume also that $F$ and $D_xF$ are continuous:
\begin{align}
&F\in C(U\times V,Z),\\
&D_xF\in C(U\times V,L(X,Z)), \quad \text{where}~ L(X,Z)\nonumber\\
&\text{denotes the Banach space of bounded linear operators}\nonumber\\
&\text{from $X$ into $Z$ endowed with the operator norm.}\nonumber
\end{align}
Then there is a neighborhood $U_1\times V_1 \subset U\times V$ of $(x_0,y_0)$ and a map ${f:V_1\to U_1 \subset X}$ such
that 
\begin{align}
&f(y_0) = x_0,\\
&F(f(y),y) =0 \quad \text{for all $y\in V_1.$}\nonumber
\end{align}
Furthermore, $f \in C(V_1,X)$ and every solution to \eqref{eq1:2july12} in $U_1\times V_1$ is of the 
form $(f(y),y)$.
Finally, if $F$ is $k$-times differentiable, then $f$ is $k$-times differentiable.
\end{theorem}

\subsection{Elliptic PDE tools}

Here we assemble some useful tools for working with nonlinear elliptic partial differential equations.
Throughout this section we will assume that $\cM$ is a closed manifold with a smooth SPD
metric $g_{ab}$ and that $\Delta$ is the associated Laplace-Beltrami operator.

\subsubsection{Maximum Principle}

In this section we present a version of the maximum principle on closed manifolds.
The following result is
well-known, but we present it here for completeness.

\begin{theorem}\label{maxp}
 Let $u \in C^2(\cM)$.
Then if 
\begin{align}\label{eq4:8july12}
\Delta u \ge 0 \quad \text{or} \quad \Delta u = 0 \quad \text{or} \quad \Delta u \le 0,
\end{align}
then $u$ must be a constant.
In particular, the problem
$$\Delta u = f(x,u),$$
has no solution if $f(x,u) \ge 0$ or $f(x,u) \le 0$ unless $f(x,u) \equiv 0$.
\end{theorem}
\begin{proof}
See \cite{PW67} for a proof.
\end{proof}

\subsubsection{Method of Sub- and Super-Solutions}\label{subsup}

Here we present a theorem that provides a method to solve an elliptic problem of the form
\begin{align}
Lu = f(x,u),
\end{align}
where
\begin{align}\label{eq1:15aug12}
Lu = -\Delta u + c(x)u,\quad \text{$c(x)\in C(\cM \times \mathbb{R})$ , \hspace{3mm} $c(x) > 0$}
\end{align}
and the function $f(x,y)$ is nonlinear in the variable $y$.
\begin{theorem}
Suppose that
$f:\cM \times \mathbb{R}^+\to \mathbb{R}$
is in $C^{k}(\cM \times\mathbb{R}^+)$.
Let $L$ be of the
form \eqref{eq1:15aug12} and suppose that there exist functions
$u_-:\cM \to\mathbb{R}$ and $u_+:\cM \to \mathbb{R}$ 
such that the following hold:
\begin{enumerate}
\item $u_-,u_+ \in C^{k}(\cM),$ 
\item $0<u_-(x) \le u_+(x) \hspace{3mm}\forall x\in \cM,$ 
\item $ Lu_- \le f(x,u_-),$
\item $ Lu_+ \ge f(x,u_+).$
\end{enumerate}
Then there exists a solution $u$ to 
\begin{align}
\label{eq2:15aug12}
Lu &= f(x,u) \hspace{3mm} \text{on $\cM$,}
\end{align}
such that
\begin{itemize}
\item[(i)] $u\in C^{k}(\cM)$,
\item[(ii)] $u_-(x)\le u(x) \le u_+(x).$
\end{itemize}
\end{theorem}  
\begin{proof}
See \cite{JI95} for a proof.
\end{proof}

\subsubsection{Fredholm Properties and Liapunov-Schmidt Decompositions for \\Elliptic Operators}\label{sec1:24july12}

In this appendix we discuss the Fredholm properties of linear elliptic operators on a closed manifold.
We use these properties to form Liapunov-Schmidt decompositions for a given elliptic 
operator $L$ between certain Banach spaces.
The following treatment is taken from \cite{HK04}.

Let $u\in C^{2,\alpha}(\cM)$ and define the elliptic operator $L:C^{2,\alpha}(\cM) \to C^{0,\alpha}(\cM)$ by
\begin{align}\label{eq1:22july12}
Lu = -\sum_{i,j=1}^n (a_{ij}(x)u_{x_i})_{x_j} +\sum_{i=1}^nb_i(x)u_{x_i}+c(x)u,
\end{align}
where $a_{ij},b_i$ and $c$ are smooth, bounded coefficients where $a_{ij} = a_{ji}$.
We also assume that the $a_{ij}$ satisfy the standard elliptic property
$$
\sum_{i,j=1}^n a_{ij} \xi_i \xi_j \ge d\|\xi\|^2,
$$
where $d>0$ is constant and $\|\cdot\|$ is the Euclidean norm on $\mathbb{R}^n$.

The operator \eqref{eq1:22july12} has an associated bilinear form
\begin{align}
B(u,u) = \langle Lu, u \rangle = \langle u,L^*u\rangle,
\end{align}
where $\langle \cdot,\cdot\rangle$ is the $L^2(\cM)$ inner product and
$L^*$ is the $L^2$-adjoint defined by
\begin{align}\label{eq4:23july12}
L^* u = -\sum_{i,j=1}^n (a_{ij}(x)u_{x_i})_{x_j} -\sum_{i=1}^n(b_i(x)u)_{x_i}+c(x)u.
\end{align}
Using the bilinear form $B(u,u)$, the elliptic operator \eqref{eq1:22july12} defines an
elliptic \\
operator 
\begin{align}
L:L^2(\cM) \to L^2(\cM), \quad \text{with domain of definition $D(L) = H^2(\cM)$.}
\end{align}
It is a standard argument in linear elliptic PDE to show that there exists a $c>0$ such the operator
$L+cI:H^2(\cM) \to L^2(\cM)$ is bounded and bijective.
In particular, one
shows that there exists a $c>0$ such that the associated bilinear form
$B(u,u) + c\|u\|_2$ is coercive and then applies the Lax-Milgram Theorem to conclude
that there exists a unique weak solution $u\in H^1(\cM)$ to 
$$
Lu-cu = f \quad \text{for every $f\in L^2(\cM)$}.
$$  
Standard elliptic regularity theory implies that $u\in H^2(\cM)$ and 
the norm $\|\cdot\|_{2,2}$ makes $D(L)$ a Hilbert space.
An application of the
Open Mapping Theorem (Bounded Inverse Theorem) then implies that 
$$(L+cI)^{-1}:L^2(\cM) \to D(L),$$
is continuous.
This implies that the operator $L+cI$ is closed and that the operator
$(L+cI)-cI  = L$ is closed.
In addition, the operator
$$
K_c =(L+cI)^{-1} \in L(L^2(\cM),L^2(\cM)) \quad \text{is compact}
$$
given that the embedding $H^2(\cM) \subset L^2(\cM)$ is compact.
For $f\in L^2(\cM)$, we have the equivalence
\begin{align}\label{eq1:23july12}
&Lu = f, \quad u \in  H^2(\cM) \Leftrightarrow \\
&u-cK_cu = K_cf, \quad u \in L^2(\cM).
\end{align}
Riesz-Schauder theory implies that $(I-cK_c)$ is a Fredholm operator and the
equivalence \eqref{eq1:23july12} implies that $L$ is a Fredholm operator.

Because $L$ is a Fredholm operator of index zero, we have that
$R(L)$ is closed.
Therefore we may write 
$$
L^2(\cM) = R(L) \oplus Z_0,
$$
where $Z_0 = R(L)^{\perp}$ is the orthogonal complement with respect
to the $L^2$-inner product.
Because $D(L)$ is dense in $L^2(\cM)$ and $L$ is closed, may apply the Closed
Range Theorem to conclude that 
\begin{align}\label{eq2:23july12}
R(L) = \{f\in L^2(\cM)~|~\langle f,u\rangle = 0 \quad \text{for all}~~u \in  N(L^*)\}
\end{align}
and that $Z_0 = N(L^*)$, where $L^*:L^2(\cM) \to L^2(\cM)$ is induced by \eqref{eq4:23july12}.
Therefore
$$
L^2(\cM) = R(L)\oplus N(L^*),
$$
and if $D(L^*) = H^2(\cM)$, the above arguments imply that $L^*$ is Fredholm
operator.
So we have the following decomposition of the codomain of $L^*$:
\begin{align}\label{eq5:23july12}
L^2(\cM) = R(L^*)\oplus N(L).
\end{align}
Finally, given that $N(L) \subset D(L)= H^2(\cM)  \subset L^2(\cM)$, the decomposition
\eqref{eq5:23july12} allows us to obtain the following Liapunov-Schmidt decomposition
for the linear problem $L:H^2(\cM) \to L^2(\cM)$:
\begin{align}
H^2(\cM)&= N(L) \oplus (R(L^*)\cap H^2(\cM)), \label{eq6:23july12}\\
L^2(\cM) &= R(L)\oplus N(L^*)\label{eq1:27aug12}.
\end{align}

Now we observe that the Fredholm properties of linear elliptic operators
derived on Hilbert spaces hold for subspaces that are only Banach spaces.
We then use
these Fredholm properties to derive Liapunov-Schmidt decompositions for these Banach spaces.

Suppose that the Banach space $Z\subset L^2(\cM)$ is continuously embedded and that the domain
of definition $X \subset Z$ with a given norm is a Banach space that satisfies the following conditions:
\begin{align}\label{eq7:23july12}
&L:X \to Z \quad \text{ is continuous ,} \\
&Lu = f \quad \text{for $u\in D(L)= H^2(\cM), f\in Z \Rightarrow u \in X$.}\nonumber
\end{align}
Equation \eqref{eq7:23july12} is an elliptic regularity condition and is satisfied for a variety of spaces, most notably
$X = W^{2,p}(\cM), Z = L^p(\cM)$ and $X = C^{2,\alpha}(\cM), Z = C^{0,\alpha}(\cM)$ with the standard norms.
Then for $X$ and $Z$ satisfying \eqref{eq6:23july12} and \eqref{eq7:23july12} we have
that
\begin{align}\label{eq8:23july12}
&N(L) = N(L|_Z) \subset X, \quad \text{and}\\
&R(L) \cap Z = R(L|_Z) \quad \text{is closed in $Z$,}
\end{align}
given that $Z\subset L^2(\cM)$ is continuously embedded and $R(L)$ is closed in 
$L^2(\cM)$.
The ellipticity property \eqref{eq7:23july12} also holds for the 
adjoint $L^*$ and implies that 
$$N(L^*) \subset X, \quad \text{where $D(L^*) =D(L) = X$}.$$
Applying the decomposition \eqref{eq1:27aug12}, we may write any $z \in Z$ as
\begin{align}\label{eq9:23july12}
&z = Lu+u^*,\quad \text{where $u\in D(L), u^* \in N(L^*)$},\\
&Lu = z-u^* \in Z \Rightarrow u \in X, \quad \text{therefore} \nonumber\\
&Z = R(L|_Z)\oplus N(L^*). \nonumber
\end{align}
Finally, we have that $\text{dim}N(L|_Z) = \text{dim}N(L) = \text{dim}N(L^*)$
and that 
\begin{align}
L:X \to Z, ~~X = D(L|_Z), ~~\text{is a Fredholm operator of index zero.}
\end{align}
 The decomposition \eqref{eq6:23july12} then implies that
\begin{align}\label{eq10:23july12}
X = N(L|_Z)\oplus (R(L^*)\cap X),
\end{align}
and so \eqref{eq9:23july12} and \eqref{eq10:23july12} constitute a Liapunov-Schmidt
decomposition of the spaces $X$ and $Z$ with respect to a given linear, elliptic operator
$L$.
\begin{remark}
As noted in \cite{HK04}, we may regard the spaces $W^{2,p}(\cM) \subset L^p(\cM) \subset L^2(\cM)$ for $p>2$, and we can then apply the above discussion to conclude that a linear elliptic operator $L:W^{2,p}(\cM) \to L^p(\cM)$ is
Fredholm and use this fact to obtain a Liapunov-Schmidt decomposition of $X= W^{2,p}(\cM)$ and
$Z = L^p(\cM)$.
Similarly, $C^{2,\alpha}(\cM) \subset C^{0,\alpha}(\cM)\subset L^2(\cM)$ for $\alpha\in (0,1)$, so 
$L:C^{2,\alpha}(\cM)\to C^{0,\alpha}(\cM)$ is Fredholm and we may also obtain a Liapunov-Schmidt decomposition of $X=C^{2,\alpha}(\cM)$ and $Z= C^{0,\alpha}(\cM)$
using \eqref{eq9:23july12} and \eqref{eq10:23july12}.
\end{remark}

\section*{Acknowledgments}\label{sec:app}

The authors wish to thank Niall O'Murchadha for a number
of helpful comments and key insights regarding the manuscript, 
as well as his overall encouragement and enthusiasm for this work.


\bibliographystyle{abbrv}
\bibliography{Mjh,Caleb,Caleb2,Caleb3}

\vspace*{0.5cm}

\end{document}